\documentclass[12pt]{article}
\usepackage{amsmath, dsfont, amssymb, amsthm, amscd, caption,  color}
\usepackage{graphicx,subfigure}
\usepackage[margin=1in]{geometry}

\usepackage{times,soul}
\usepackage{bm}
\usepackage{natbib}

\usepackage[plain,noend]{algorithm2e}
\usepackage{setspace}
\usepackage{tikz}
\usepackage{url}

\newcommand{\utilde}{ }

\newcommand{\bpm}{\begin{pmatrix}}
\newcommand{\epm}{\end{pmatrix}}
\def\T{{ \mathrm{\scriptscriptstyle T} }}

\newtheorem{theorem}{Theorem}
\newtheorem{lemma}{Lemma}
\newtheorem{corollary}{Corollary}

\newtheorem{proposition}{Proposition}
\theoremstyle{definition}
\newtheorem{definition}{Definition}
\newtheorem{assumption}{Assumption}

\newtheorem{example}{Example}

\theoremstyle{remark}
\newtheorem{remark}{Remark}

\begin{document}

\title{Selection and Estimation for Mixed Graphical Models}
\author{Shizhe Chen, Daniela Witten, and Ali Shojaie \\
 Department of Biostatistics, University of Washington, Box 357232,\\
  Seattle, WA 98195-7232\\ 
 }

\date{August 1st, 2014}
\maketitle

\begin{abstract}
We consider the problem of estimating the parameters in a pairwise graphical model in which the distribution of each node, conditioned on the others, may have a different parametric form. In particular, we assume that each node's conditional distribution is in the exponential family. We identify restrictions on the parameter space required for the existence of a well-defined joint density, and establish the consistency of the neighbourhood selection approach for graph reconstruction in high dimensions when the true underlying graph is sparse. Motivated by our theoretical results, we investigate the selection of edges between nodes whose conditional distributions take different parametric forms, and show that efficiency can be gained if edge estimates obtained from the regressions of particular nodes are used to reconstruct the graph. These results are illustrated with examples of Gaussian, Bernoulli, Poisson and exponential distributions. Our theoretical findings are corroborated by evidence from simulation studies. 
\end{abstract}

\textbf{Keywords:} compatibility; conditional likelihood;  exponential family; high-dimensionality;  model selection consistency; neighbourhood selection; pairwise Markov random field.

\newpage

\section{Introduction}\label{intro}

In this paper, we consider the task of learning the structure of an undirected graphical model encoding pairwise conditional dependence relationships among random variables. Specifically, suppose that we have $p$ random variables represented as nodes of the graph $G=(V,E)$, with the vertex set $V= \{1, \ldots, p \}$ and the edge set $E \subseteq  V \times V $. An edge in the graph indicates a pair of random variables that are conditionally dependent given all other variables. The problem of reconstructing the  graph from a set of $n$ observations has attracted a lot of interest in recent years, especially when $p>n$ and $p(p-1)/2$ edges must be estimated from $n$ observations.

Many authors have studied the estimation of high-dimensional undirected graphical models in the setting where {the distribution of each node, conditioned on all other nodes, has the same parametric form}.
 In particular, Gaussian graphical models have been studied extensively (see e.g. \citealt{meinshausen2006}, \citealt{yuan2007}, \citealt{friedman2008}, \citealt{rothman2008}, \citealt{wainwright2008}, \citealt{peng2009}, \citealt{ravikumar2011}), 
and have been generalized to account for non-normality and outliers (see e.g. \citealt{miyamura2006},  \citealt{finegold2011}, \citealt{vogel2011}, \citealt{sun2012}). 
  Others have considered the setting in which all node-conditional distributions are   Bernoulli  (\citealt{lee2006},  \citealt{hofling2009}, and \citealt{ravikumar2010}), multinomial  (\citealt{jalali2011}), Poisson \citep{allen2012}, or any univariate distribution in the exponential family \citep{yang2012,yang2013}.

In this paper, we seek to estimate a graphical model in which the variables are of different types. 
Here, the type of a node refers to the parametric form of its distribution, conditioned on all other nodes.
 For instance, the variables might include DNA nucleotides, taking binary values, and gene expression measured using RNA-sequencing, taking non-negative integer values. 
We could model the first set of nodes as Bernoulli, which means that each of their distributions, conditional on the other nodes, is Bernoulli; similarly, we could model the second set as  Poisson.
 We assume that the type of each node is known a priori, and refer to this setup as a mixed graphical model.

 In the low-dimensional setting, \citet{lauritzen1996} studied a special case of the mixed graphical model, known as the conditional Gaussian model, in which each node is either Gaussian or Bernoulli. More recent work has focused on the high-dimensional setting.  \citet{lee2012} proposed two algorithms for reconstructing conditional Gaussian models using a group lasso penalty. \citet{cheng2013} modified this approach by using a weighted $\ell_1$-penalty. 
 
  A related line of research  considers semi-parametric or non-parametric approaches for estimating conditional dependence relationships  (\citealt{liu2009},  \citealt{xue2012}, \citealt{fellinghauer2013}, \citealt{voorman2014}), among which \citet{fellinghauer2013} is specifically proposed for mixed graphical models. However, despite their flexibility, these non-parametric methods are often less efficient than their parametric counterparts, if the type of each node is known.

In this paper, we propose an estimator and develop theory for the parametric mixed graphical model, under a much more general setting than existing approaches (e.g. \citealt{lee2012}). We allow the conditional distribution of each node to belong to the exponential family. Unlike \citet{yang2012}, nodes may be of different types. For instance, within a single graph, some nodes may be Bernoulli, some may be Poisson, and some may be exponential. 

In parallel efforts, \citet{yang2014} recently presented general results on strong compatibility for mixed graphical models for which the node-conditional distributions belong to the exponential family, and for which the graph contains only two types of nodes. We instead consider the setting where the graph can contain more than two types of nodes, and provide specific requirements for strong compatibility for some common distributions.

\section{A Model for Mixed Data} \label{model}
\subsection{Conditionally-Specified Models for Mixed Data}\label{condmodel}
We consider the pairwise graphical model \citep{wainwright2006}, which takes the form
\begin{equation}
p(x) \propto \exp \left\{ \sum\limits_{s=1}^p f_s(x_s) + \sum\limits_{s = 2 }^p \sum\limits_{t < s }  f_{ts} (x_s,x_t) \right\},
\label{PGM}
\end{equation}
where ${x}={(x_1, ..., x_p)}^\T$ and $f_{ts}= 0$ for $\{t,s\} \notin E$. Here, $f_s(x_s)$ is the node potential function, and $f_{st} (x_s,x_t)$ the edge potential function. We further simplify the pairwise interactions by assuming that $f_{st} (x_s,x_t) = \theta_{st} x_s x_t = \theta_{ts} x_s x_t$, so that we can write the parameters associated with edges in a symmetric square matrix $\Theta = (\theta_{st})_{p \times p}$ where the diagonal elements equal zero. The joint density can then be written as
\begin{equation}
p({x}) = \exp \left\{ \sum\limits_{s=1}^p  f_s(x_s) + \frac{1}{2} \sum\limits_{s=1}^{p}\sum\limits_{t \neq s} \theta_{ts} x_s x_t - A(\Theta, \alpha) \right\},
\label{joint}
\end{equation}
where $A(\Theta,\alpha)$ is the log-partition function, a function of  $\Theta$ and $\alpha$. Here $\alpha$ is a $ K \times p$ matrix of parameters involved in the node potential functions: that is, $f_s(x_s)$ involves $\alpha_s$, the $s$th column of $\alpha$.  $K$ is some known integer.  For $\{s,t\} \notin E$, the edge potentials satisfy $\theta_{st}=\theta_{ts}=0$. We define the neighbours of the $s$th node as $N(x_s)=\{t: \theta_{st}=\theta_{ts}\neq 0 \}$.

In principle, given a parametric form for the joint density \eqref{joint},  we can estimate the conditional dependence relationships among the $p$ variables, and hence the edges in the graph. But this approach requires the calculation of the log-partition function $A(\Theta, \alpha)$, which is often intractable. To overcome this, we instead use the framework of conditionally-specified models \citep{besag1974}: we specify the distribution of each node conditional on the others, and then combine the $p$ conditional distributions to form a single graphical model. This approach has been widely used in estimating high-dimensional graphical models where all nodes are of the same type \citep{meinshausen2006,ravikumar2010,allen2012,yang2012}. However, as we will discuss in Section~\ref{compatibility}, a conditionally-specified model may not correspond to a valid joint distribution.

Define ${x}_{-s}={(x_1,...,x_{s-1}, x_{s+1},...,x_p)}^\T$. We consider conditional densities of the form
\begin{equation}
p(x_s \mid {x}_{-s}) =  \exp \left\{ f_s(x_s)  + \sum\limits_{ t\neq s} \theta_{ts} x_t x_s -  D_s(\eta_s) \right\},
\label{cond}
\end{equation}
where $\eta_s=\eta_s(\Theta_s, x_{-s}, \alpha_s)$ is a function of $\alpha_s$, $x_{-s}$, and $\Theta_{s}$,  and $\Theta_s$ is the $s$th column of $\Theta$ without the diagonal element.  
Suppose $f_s(x_s)= \alpha_{1s} x_s + \alpha_{2s} x_s^2/2+ \sum_{k = 3}^{K} \alpha_{ks} B_{ks}(x_s)$, where $\alpha_{ks}$ is a parameter, which could be 0, and $B_{ks}(x_s)$ is a known function for $k=3,\ldots,K$.  Under this assumption, \eqref{cond} belongs to the exponential family. 

The assumed form of $f_s(x_s)$ is quite general. We now consider some special cases of \eqref{cond} corresponding to  commonly-used distributions in the exponential family, for which $f_s(x_s)$ takes a very simple form. In the following examples, we assume that  $\eta_s(\Theta_s, x_{-s}, \alpha_s)=\alpha_{1s} + \sum_{t: \; t \neq s} \theta_{ts} x_t$.

\begin{example}
 The conditional density is  Gaussian and $\alpha_{2s}=-1$:
\begin{equation}
p(x_s \mid {x}_{-s})=\exp \left\{ -\frac{1}{2} x_s^2 + \eta_s x_s - \frac{1}{2} \eta_s^2 -\frac{1}{2}\log(2\pi)   \right\}, \quad x_s \in \mathcal{R},
\label{cond:Gaussian}
\end{equation}
where $f_s(x_s) = \alpha_{1s} x_{s} - x_s^2/2 $ and $D_s(\eta_s)= \eta_s^2/2+\log(2\pi)/2$.\end{example}

\begin{example}
 The conditional density is Bernoulli. Instead of coding $x_s$ as $\{0,1\}$, we code $x_s$ as $\{-1,1\}$. This yields the conditional density
\begin{equation}
p(x_s\mid {x}_{-s})=\exp \left\{ \eta_s x_s   - D_s(\eta_s) \right\}, \quad x_s \in \{ -1, 1\},
\label{cond:Bernoulli}
\end{equation}
where $f_s(x_s)=\alpha_{1s} x_s$ and $D_s(\eta_s)=\log\{ \exp(\eta_s)+\exp(-\eta_s )\}$.
\end{example}

\begin{example}
 The conditional density is  Poisson:
\begin{equation}
p(x_s \mid  {x}_{-s})=\exp \left\{ \eta_s x_s  -\log(x_s !)- D_s(\eta_s)\right\}, \quad x_s \in \{0, 1, \ldots \} ,  
\label{cond:Poisson}
\end{equation}
where $f_s(x_s)=\alpha_{1s} x_s-\log(x_s!)$ and $D_s(\eta_s)=\exp(\eta_s)$.
\end{example}

\begin{example}
 The conditional density is  exponential:
\begin{equation}
p(x_s \mid  {x}_{-s})=\exp \left\{ \eta_s x_s  - D_s(\eta_s)\right\}, \quad x_s \in \mathcal{R}^{+},
\label{cond:exp}
\end{equation}
where $f_s(x_s)=\alpha_{1s} x_s$ and $D_s(\eta_s)=-\log(-\eta_s)$.
\end{example}
These four examples  have been studied in the context of  conditionally-specified graphical models in which all nodes are of the same type (\citealt{besag1974}, \citealt{meinshausen2006}, \citealt{ravikumar2010}, \citealt{allen2012},  \citealt{yang2012}). 

 In what follows, we will consider the conditionally-specified mixed graphical model, with conditional distributions given by (\ref{cond}), in which each node can be of a different type. This class of mixed graphical models is not closed under marginalization: for instance, given a graph composed of Gaussian and Bernoulli nodes, integrating out the Bernoulli nodes leads to a conditional density that is a mixture of Gaussians, which does not belong to the exponential family.

\subsection{Compatibility of Conditionally-Specified Models}\label{compatibility}
Under what circumstances does the conditionally-specified model with node-conditional distributions given in \eqref{cond} correspond to a well-defined joint distribution?   
 We first adapt and restate a definition from \citet{wang2008}, which applies to any conditional density.
\begin{definition}
A non-negative function $g$ is capable of generating a conditional density function $p(y \mid  {x})$ if
\begin{equation*}
p(y \mid  {x})=\frac{g(y, {x})}{\int g(y, {x}) dy}.
\end{equation*}
Two conditional densities are said to be compatible if there exists a function $g$ that  is capable of generating both conditional densities. When $g$ is a density, the conditional densities are called strongly compatible.
\label{compat:defn}
\end{definition}

The following proposition relates Definition~\ref{compat:defn} to the  conditional density  in \eqref{cond}. Its proof, and those of other statements in this paper, are available in the Supplementary Material.
\begin{proposition}
Let $ {x}={(x_1,...,x_p)}^\T$ be a random vector. Suppose that for each $x_s$, the conditional density takes the form of \eqref{cond}.
 If $\theta_{st}=\theta_{ts}$, then  the conditional densities are compatible. Furthermore, any function $g$ that is capable of generating the conditional densities is of the form
\begin{equation}
g( {x})\propto \exp \left\{ \sum\limits_{s=1}^p f_s(x_s) + \frac{1}{2}\sum\limits_{s=1}^{p}\sum\limits_{t\neq s} \theta_{ts} x_s x_t \right\}.
\label{g}
\end{equation}
 \label{compat:prop}
\end{proposition}

Under the conditions of Proposition~\ref{compat:prop}, if we further assume that $g$ in \eqref{g} is integrable, then by Definition~\ref{compat:defn}, the conditional densities of the form (\ref{cond}) are strongly compatible. 
 Proposition~\ref{compat:prop} indicates that, provided that \eqref{joint} is a  valid joint distribution, we can arrive at it via the conditional densities in \eqref{cond}.  This justifies the conditionally-specified modeling approach taken in this paper.  Proposition~\ref{compat:prop} is closely related to Section 4$\cdot$3 in \citet{besag1974} and Proposition 1 in \citet{yang2012}, with small modifications. More general theory is developed in \citet{wang2008}.

We now return to the four examples \eqref{cond:Gaussian}--\eqref{cond:exp}. Lemma~\ref{lmm:compat} summarizes the  conditions under which a conditionally-specified model with {non-degenerate} conditional distributions of the form \eqref{cond:Gaussian}--\eqref{cond:exp} leads to a valid joint distribution. 
\begin{lemma} 
If $\theta_{st}=\theta_{ts}$, the subset of conditions with a dagger ($\dagger$) in Table~\ref{tab1} is necessary and sufficient for the conditional densities in \eqref{cond:Gaussian}--\eqref{cond:exp} to be compatible. Moreover, the complete set of conditions in Table~\ref{tab1} is necessary and sufficient for the conditional densities in \eqref{cond:Gaussian}--\eqref{cond:exp} to be strongly compatible.
\label{lmm:compat}
\end{lemma}

 {To simplify the presentation of the conditions for the Gaussian nodes, in Table~\ref{tab1} it is assumed} 
 that $J$ is the index set of the Gaussian nodes. Without loss of generality, we further assume that the nodes are ordered such that $J=\{1,\ldots,m\}$, and define 
 \begin{equation}
 \Theta_{JJ}= \bpm \alpha_{2 1} & \theta_{1 2} & \cdots & \theta_{1 m} \\
\theta_{2 1}   & \alpha_{2 2} &\cdots &\theta_{2 m} \\
\vdots & \vdots & \ddots & \vdots \\
\theta_{m 1} & \theta_{m 2} & \cdots & \alpha_{2 m}
  \epm.
 \label{Thetajj}
\end{equation}

\begin{table}
\def~{\hphantom{0}}
\caption{Restrictions on the parameter space required for compatibility or strong compatibility of the conditional densities in \eqref{cond:Gaussian}--\eqref{cond:exp}}
\begin{center}
\begin{tabular}{ ccccc }
\multicolumn{1}{c}{} &  \multicolumn{1}{c}{Gaussian}  & \multicolumn{1}{c}{Poisson} &
 \multicolumn{1}{c}{Exponential} &  \multicolumn{1}{c}{Bernoulli}  \\
Gaussian &  $\Theta_{JJ} \prec 0$  & $\theta_{ts}=0$ &$\theta_{ts}=0^\dagger$ & $\theta_{ts} \in \mathcal{R}^\dagger$ \\
\multicolumn{1}{c}{Poisson} &  & $\theta_{ts}\leq 0$& {$\theta_{ts}\leq 0^\dagger$} & $\theta_{ts} \in \mathcal{R}^\dagger$ \\
\multicolumn{1}{c}{Exponential} & \multicolumn{1}{r}{} & & $ \theta_{ts}\leq 0^\dagger$  & $ \sum_{s \in I}|\theta_{st}| < - \alpha_{1t}^\dagger$\\
\multicolumn{1}{c}{Bernoulli} &  \multicolumn{2}{r}{ } & &$\theta_{ts} \in \mathcal{R}^\dagger$\\
\end{tabular}
\label{tab1}
\end{center}
The column specifies the type of the $s$th node, and the row specifies the type of the  $t$th node. Conditions marked with a dagger ($\dagger$) are necessary and sufficient for the conditional densities in \eqref{cond:Gaussian}--\eqref{cond:exp} to be compatible, and the complete set of conditions is necessary and sufficient for the conditional densities to be strongly compatible. For compatibility to hold for a Gaussian node $x_s$, $\alpha_{2s}<0 $ is also required. Here $\Theta_{JJ}$ is as defined in \eqref{Thetajj}, and  $I$ denotes the set of Bernoulli nodes. 
\end{table}

Table 1 reveals the set of restrictions on the parameter space that must hold in order for the conditional densities in \eqref{cond:Gaussian}--\eqref{cond:exp} to be compatible or strongly compatible. The diagonal entries of this table were previously studied in \cite{besag1974}. In general, strong compatibility imposes more restrictions on the parameter space than compatibility. For instance, compatibility does not place any restrictions on edges between two Poisson nodes, but for strong compatibility to hold, the edge potentials must be negative. Compatibility and strong compatibility even restrict the relationships that can be modeled using the conditional densities \eqref{cond:Gaussian}--\eqref{cond:exp}: for instance, no edges are possible between Gaussian and exponential nodes, or between Gaussian and Poisson nodes.

 To summarize, given conditional densities of the form \eqref{cond:Gaussian}--\eqref{cond:exp}, existence of a joint density imposes substantial constraints on the parameter space, and thus limits the flexibility of the corresponding graph. However, we will see in Section~\ref{nojd} that it is possible to consistently estimate the structure of a graph even when the requirements for compatibility or strong compatibility are violated, i.e., even in the absence of a joint density.

While Table~\ref{tab1} only examines conditionally-specified models composed of the conditional densities in \eqref{cond:Gaussian}--\eqref{cond:exp}, the estimator proposed in Section~\ref{pf} and the theory developed in Sections~\ref{theory} and \ref{nojd} apply to other types of conditional densities of the form \eqref{cond}.

\section{Estimation Via Neighbourhood Selection}\label{pf}
\subsection{Estimation}\label{estimation}
We now present a neighbourhood selection approach for recovering the structure of a mixed graphical model, by maximizing penalized conditional likelihoods node-by-node.
 A similar approach has been studied in the setting where all nodes in the graph are of the same type \citep{meinshausen2006,ravikumar2010,allen2012,yang2012}.

Recall from Section~\ref{condmodel} that  $f_s(x_s)= \alpha_{1s} x_s + \alpha_{2s} x_s^2/2+ \sum_{k = 3}^{K} \alpha_{ks} B_{ks}(x_s)$. We now simplify the problem by assuming that $\alpha_{ks}$ is known, and possibly zero, for $k\geq 2$.  Let $X$ denote an $n \times p$ data matrix, with the $i$th row given by $x^{(i)}$. 
From now on, we use an asterisk to denote the  true parameter values. We estimate $\Theta^*_s$ and $\alpha^*_{1s}$, the parameters for the $s$th node,  as
\begin{equation}  	
\underset{{\Theta_s \in  \mathcal{R}^{p-1},\ \alpha_{1s} \in \mathcal{R}}}{\text{{arg min}}}  \quad -\ell_s (\Theta_s, \alpha_{1s};  {X} ) +\lambda_n \| \Theta_s \|_1,
\label{pl}
\end{equation}
where
$\ell_s (\Theta_s, \alpha_{1s} ;  {X} )=  \sum\limits_{i=1}^{n} \log p(x_s^{(i)}\mid  {x}^{(i)}_{-s})/n$; recall that the conditional density $p(x_s^{(i)}\mid  {x}^{(i)}_{-s})$ is defined in \eqref{cond}.  Finally, we define the estimated neighbourhood of $x_s$ to be $\hat{N}(x_s)= \{t: \hat{\theta}_{ts}\neq 0\}$, where $\hat{\Theta}_s$ solves \eqref{pl}, and $\hat\theta_{ts}$ is the element corresponding to an edge with the $t$th node.

In practice,  to avoid a situation in which variables of different types are on different scales, we may wish to modify \eqref{pl} in order to allow a different weight for the $\ell_1$-penalty on each coefficient.  We define a weight vector $w$ equal to the empirical standard errors of the corresponding variables: $w={(\hat{\sigma}_1, ... ,\hat{\sigma}_{s-1}, \hat{\sigma}_{s+1}, ..., \hat{\sigma}_{p})}^\T$. Then \eqref{pl} can be replaced with 
\begin{equation}  	
\underset{{\Theta_s \in  \mathcal{R}^{p-1},\ \alpha_{1s} \in \mathcal{R}}}{\text{arg min}}  \quad -\ell_s (\Theta_s, \alpha_{1s};  {X} ) +\lambda_n \| {\text{diag}(w)} \Theta_s \|_1.
\label{pl.weighted}
\end{equation}
The analysis in Sections~\ref{theory} and \ref{nojd} uses \eqref{pl} for simplicity, but could be  generalized to  \eqref{pl.weighted} with additional bookkeeping. Both  \eqref{pl} and \eqref{pl.weighted}  can be easily solved  (see e.g. \citealt{friedman2010}).

In the joint density \eqref{joint}, the parameter matrix $\Theta$ is symmetric, i.e., $\theta_{st}=\theta_{ts}$, but the neighbourhood selection method does not guarantee symmetric estimates: for instance, it could happen that $\hat{\theta}_{st}=0$ but $\hat{\theta}_{ts}\neq 0$. Our analysis in Section~\ref{selection} shows that we can exploit the asymmetry in $\hat{\theta}_{st}$ and $\hat{\theta}_{ts}$ when $x_s$ and $x_t$ are of different types, in order to obtain more efficient edge estimates.

\subsection{Tuning}\label{tuning}
In order to select the value of  the tuning parameter $\lambda_n$ in (\ref{pl}), we use the Bayesian information criterion (\citealt{zou2007}, \citealt{peng2009}, \citealt{voorman2014}), which takes the form
\begin{equation}
\textsc{bic}_s({\lambda_n}) = - 2n \ell_s (\hat{\Theta}_s, \hat{\alpha}_{1s}; X) + \log(n)  \| \hat{\Theta}_s \|_0,
\end{equation} 
{where $\| \hat{\Theta}_s \|_0$ is the number of non-zero elements in $\hat{\Theta}_{s}$ for a given value of $\lambda_n$. }
We allow a different value of $\lambda_n$ for each node type. For instance, to select $\lambda_n$ for the Poisson nodes, we choose the value of $\lambda_n$ such that $\textsc{bic}_s({\lambda_n})$, summed over the Poisson nodes, is minimized. We evaluate the performance of this approach for tuning parameter selection in Section~\ref{TP}.

\section{Neighbourhood Recovery and Selection With Strongly Compatible Conditional Distributions}\label{theory}
\subsection{Neighbourhood Recovery } \label{subsec:jd}
In this subsection we show that if the conditional distributions in \eqref{cond} are strongly compatible, as they will be under conditions discussed in Section~\ref{compatibility}, then under some additional assumptions, the true neighbourhood of each node is consistently selected using the neighbourhood selection approach proposed in Section~\ref{estimation}. Here we rely heavily on results from \citet{yang2012}, who consider a related problem in which all nodes are of the same type.

In the following discussion, we assume that $p>n$ for simplicity. For any $s$, let $\Delta_s$ denote the set of indices for elements of $(\Theta_s^\T, \alpha_{1s})^\T$ that correspond to non-neighbours of the $s$th node, and let $Q_s^*=-\nabla^2 \ell_s(\Theta_s^*, \alpha_{1s}^*; {X})$ be the negative Hessian of $\ell_s(\Theta_s,\alpha_{1s}; {X})$ with respect to $(\Theta_s^\T, \alpha_{1s})^\T$, evaluated at the true values of the parameters. Below we suppress the subscript $s$ for simplicity, and we remind the reader that all quantities are related to the conditional density of the   $s$th node. We express $Q^*$ in blocks: $$Q^*=\bpm Q^*_{\Delta^c \Delta^c} & Q^*_{\Delta^c \Delta} \\ Q^*_{\Delta \Delta^c} & Q^*_{\Delta \Delta} \epm. $$

\begin{assumption}
There exists a positive number $a$ such that $$\underset{l \in \Delta}{\max} \| Q^*_{l\Delta^c} (Q^*_{\Delta^c \Delta^c})^{-1}\|_1 \leq 1-a .$$
\label{irrep}
\end{assumption}
Assumption~\ref{irrep} limits the association between the neighbours and non-neighbours of the $s$th node: if the association is too high, then it is not possible to select the correct neighbourhood. This type of assumption is standard  for variable selection consistency of $\ell_1$-penalized estimators (see e.g. \citealt{meinshausen2006}, \citealt{zhao2006}, \citealt{wainwright2009}, \citealt{ravikumar2010}, \citealt{ravikumar2011},  \citealt{yang2012}, \citealt{lee2013}).

\begin{assumption}
There exists $ \Lambda_{1} > 0$ such that the smallest eigenvalue of $Q^*_{\Delta^c \Delta^c}$, $\Lambda_{\min} (Q^*_{\Delta^c \Delta^c})$, is greater than or equal to $\Lambda_{1}.$  Also, there exists $ \Lambda_{2} < \infty$ such that the largest eigenvalue of $\sum_{i=1}^{n}   {x}^{(i)}_{0} ({x}^{(i)}_{0})^\T/n$, $\Lambda_{\max}\left\{ \sum_{i=1}^{n}   {x}^{(i)}_{0} ({x}^{(i)}_{0})^\T/n \right\}$,  is less than or equal to $\Lambda_{2}$, where ${x}_{0}=(x_{-s}^\T,1)^\T$. 
\label{dep}
\end{assumption}
The lower bound in Assumption~\ref{dep} is needed to prevent singularity among the true neighbours, which would prevent neighbourhood recovery. The bound on the largest eigenvalue of the sample covariance matrix  is needed to prevent a situation where most of the variance in the data is due to a single feature. Similar assumptions are made in  \citet{zhao2006}, \citet{meinshausen2006},  \citet{wainwright2009}, \citet{ravikumar2010}, \citet{yang2012}.

\begin{assumption}
The log-partition function $D(\cdot)$ of the conditional density $p(x_s \mid x_{-s})$  is third-order differentiable, and there exist $\kappa_2$ and $\kappa_3$ such that  $| D^{''}(y)|\leq \kappa_2 $ and $|D^{'''}(y)|\leq  \kappa_3$
 for  $y \in \{y: y \in \mathcal{D}, \  |y| \leq M \delta_1 \log p\}$, where $\mathcal{D}$ is the support of $D(\cdot)$. 
\label{D}
\end{assumption}
\begin{remark}
The two quantities $\kappa_2$ and $\kappa_3$ are functions of $p$. The quantity $\delta_1$ is a constant to be chosen in Proposition~\ref{prop.e1}. The constant $M$ is a sufficiently large constant that plays a role in Assumption~\ref{tuning.range}.
\end{remark}
Assumption~\ref{D}  controls the smoothness of the log-partition function $D(\cdot)$ for conditional densities of the form \eqref{cond}. Recall from Section~\ref{condmodel} that the log-partition function of the node $x_s$ is $D(\eta_s)$, where $\eta_s$ equals  $\alpha_{1s}+\sum_{t \neq s} \theta_{ts} x_t$. To apply Assumption~\ref{D} to $D(\eta_s)$, we will need to bound $\sum_{t \neq s} \theta_{ts} x_t$, so that  $|\eta_s| \leq M \delta_1 \log(p)$.

 In order to obtain such a bound, we need another assumption.
\begin{assumption} Assume that, for $ \ t = 1,...,p$, (i) $|E(x_t)|\leq \kappa_m$, (ii) $E(x_t^2)\leq \kappa_v$, and (iii)
$$ \underset{u:|u|\leq 1}{\max} \left. \frac{\partial^2 A}{\partial \alpha^2_{1t}} \right|_{\alpha^*_{1t}+u} \leq \kappa_h, \quad  \underset{u:|u|\leq 1}{\max} \left. \frac{\partial^2 A}{\partial \alpha^2_{2t}} \right|_{\alpha^*_{2t}+u} \leq \kappa_h.$$
\label{A}
\end{assumption}

Assumption~\ref{A} controls the moments of each node, as well as the local smoothness of the log-partition function $A$ in \eqref{joint}. Given Assumption~\ref{A}, the following propositions on the marginal behaviour of random variables hold; see Propositions 3 and 4 in \citet{yang2012}.

\begin{proposition}
Define the event $$\xi_1 = \left( \displaystyle \max_{\substack{i\in \{1,...,n\}; t \in\{1,...,p\} }}|x^{(i)}_t| < \delta_1 \log p\right).$$ Assuming $p>n$,  $\text{pr}(\xi_1)\geq  1- c_1  p^{-\delta_1+2}, $ where $c_1=\exp(\kappa_m + \kappa_h/2) $.
\label{prop.e1}
\end{proposition}
\begin{proposition}
Define the event $$\xi_2 = \left[ \displaystyle \max_{\substack{ t \in\{1,...,p\} }}\left\{\frac{1}{n}\sum\limits_{i=1}^{n}(x^{(i)}_t)^2 \right\}< \delta_2 \right],$$ where $\delta_2 \geq 1$. If $\delta_2 \leq \min (2\kappa_v/3, \kappa_h+\kappa_v )$, and $n\geq 8 \kappa_h^2 \log p /\delta_2^2 $, then
$\text{pr}(\xi_2)\geq 1-\exp(-c_2 \delta_2^2 n), $ where $c_2=1/(4\kappa_h^2)$.
\label{prop.e2}
\end{proposition}

We now present three additional assumptions that relate to the node-wise regression in (\ref{pl}).

\begin{assumption}
The minimum of edge potentials related to node $x_s$, $ {\min}_{t\in N(x_s)} |\theta_{ts}|$, is larger than $10 (d+1)^{1/2} \lambda_n/\Lambda_{1} $, where $d$ is the number of neighbours of $x_s$.
\label{thetamin}
\end{assumption}

\begin{assumption} \label{tuning.range}
The tuning parameter $\lambda_n$ is in the range
\small
\begin{equation}
\left[\frac{8(2-a)}{a} \left\{\delta_2 \kappa_2 \frac{\log (2p)}{n}\right\}^{1/2},
\min\left\{ \frac{2(2-a)}{a} \kappa_2 \delta_2 M ,  \frac{a \Lambda_{1}^2 (d+1)^{-1}}{288(2-a) \kappa_2 \Lambda_{2} }, \frac{\Lambda_{1}^2  (d+1)^{-1}}{12 \Lambda_{2}\kappa_3 \delta_1 \log p} \right\}\right].
\label{tuning.range:eq}
\end{equation}
\end{assumption}

\begin{remark}
Of the three quantities in the upper bound of $\lambda_n$,  $\Lambda_{1}^2/ \{12\Lambda_{2}  (d+1)\kappa_3 \delta_1 \log p\}$ is usually the smallest because of the $\log p$ in the denominator.
\end{remark}

\begin{assumption}
The sample size $n$ is no smaller than $8 \kappa_h^2\log p/\delta_2^2$, and also the range of feasible $\lambda_n$ in Assumption~\ref{tuning.range} is non-empty, i.e.,
\begin{equation}
\label{sample:eq}
n\geq \frac{ 96^2 (2-a)^2 \Lambda^2_{2} }{a^2 \Lambda_{1}^4} (d+1)^2 \kappa_2 \kappa^2_3 \delta_1^2 \delta_2  \log (2p)  (\log p)^2.
\end{equation}
\label{n}
\end{assumption}
Assumptions~\ref{thetamin}, \ref{tuning.range}, and \ref{n} specify the minimum edge potential,  the range of the tuning parameter, and  the minimum sample size, required  for Theorem~\ref{thm} to hold, that is, for our neighbourhood selection approach \eqref{pl} to achieve model selection consistency.   Similar assumptions are made in related work \citep{yang2012}.

\begin{remark} \label{remark3}
 Suppose that  $n=\Omega \{ (d+1)^2 \log^{3+\epsilon}(p) \}$ for $\epsilon>0$,  $\lambda_n = c \{ \log(p)/n \}^{1/2}$ for some constant $c$, {and $\kappa_2$ and $\kappa_3$ are $O(1)$}.  Then Assumptions~\ref{tuning.range} and \ref{n} are satisfied asymptotically as $n$ and $p$ tend to infinity. Similar rates appear in \citet{meinshausen2006,ravikumar2010,yang2012}. 
\end{remark}

\begin{theorem} Suppose that the joint density (\ref{joint}) exists and Assumptions \ref{irrep} -- \ref{n} hold for {the $s$th node}. Then with probability at least $1-c_1 p^{-\delta_1+2}-\exp(-c_2 \delta_2^2 n)-\exp(-c_3 \delta_3 n)$, for some constants $c_1, c_2, c_3$, $\delta_2 \leq \min (2\kappa_v/3, \kappa_h+\kappa_v )$, and $\delta_3=1/(\kappa_2 \delta_2)$, the  {estimator from} (\ref{pl}) recovers the true neighbourhood of $x_s$ exactly, so that $\hat{N}(x_s)=N(x_s)$.
 \label{thm}
\end{theorem}

Theorem~\ref{thm} shows that the probability of successful recovery converges to 1 {exponentially fast with the sample size $n$}. 
 We note that the number of neighbours  $d$ appears in Assumptions~\ref{thetamin}--\ref{n}. As $d$ increases, the minimum edge potential for each neighbour increases, the upper range for $\lambda_n$ decreases, and the required sample size increases. 
 Therefore, we need the true graph $G$ to be sparse, $d=o(n)$, in order for Theorem 1 to be meaningful.

The quantities $\delta_2 \kappa_2$ and $\delta_1\kappa_3$ appear in the upper bound of $\lambda_n$ \eqref{tuning.range:eq} and the minimum sample size \eqref{sample:eq}. The fact that $\kappa_2$ and  $\delta_2$ appear together in a product implies that we can relax the restriction on $\delta_2$ if $\kappa_2$ is small. The same applies to $\delta_1$ and $\kappa_3$.

For certain types of nodes, Theorem~\ref{thm} holds with a less stringent set of assumptions. For a Gaussian node, the second-and-higher order derivatives of $D(\cdot)$ are always bounded, i.e., $\kappa_2=1$ and $\kappa_3=0$. This has profound effects on the theory, as illustrated in Corollary~\ref{col1}.
 \begin{corollary}
Suppose that the joint density (\ref{joint}) exists and Assumptions~\ref{irrep}--\ref{thetamin} hold {for a Gaussian node, $x_s$. If } 
\begin{equation*}
\lambda_n \in \left[\frac{8(2-a)}{a} \left\{\delta_2 \frac{\log (2p)}{n}\right\}^{1/2}, \frac{2(2-a)}{a} \delta_2 M \right], \quad n \geq \frac{8 \kappa_h^2\log p}{\delta_2^2},
\end{equation*} 
then with probability at least $1-\exp(-c_2 \delta_2^2 n)-\exp(-c_3 \delta_3 n)$, for some constants $c_2, c_3$, $\delta_2 \leq \min (2\kappa_v/3, \kappa_h+\kappa_v )$, and $\delta_3=1/ \delta_2$, the  {estimator from} \eqref{pl} recovers the true neighbourhood of $x_s$ exactly, so that $\hat{N}(x_s)=N(x_s)$.
\label{col1}
\end{corollary}

\subsection{Combining Neighbourhoods to Estimate the Edge Set}\label{selection}

The neighbourhood selection approach may give asymmetric estimates, in the sense that $t \in \hat{N}(x_s)$ but $ s \notin \hat{N}(x_t)$. To deal with this discrepancy, two strategies for estimating a single edge set were proposed in \citet{meinshausen2006}, and adapted in other
work: 
\begin{equation*}
\hat{E}_{\text{and}}= \left\{ (s,t): s \in \hat{N}(x_t) \ \text{and} \ t \in \hat{N}(x_s) \right\}, \quad
\hat{E}_{\text{or}}=  \left\{ (s,t): s \in \hat{N}(x_t) \ \text{or} \ t \in \hat{N}(x_s) \right\}.
\end{equation*}
When the $s$th and $t$th nodes are of the same type, there is no clear reason to prefer the edge estimate from $\hat{N}(x_s)$ over the one from $\hat{N}(x_t)$, and so the choice of the intersection rule, $\hat{E}_{\text{and}}$, versus the union rule, $\hat{E}_{\text{or}}$, is not crucial \citep{meinshausen2006}.  

When the $s$th and $t$th nodes are of different types, however, the choice of neighbourhood matters. 
 We now take a closer look at this with examples of Gaussian, Bernoulli, exponential and Poisson nodes as in \eqref{cond:Gaussian}--\eqref{cond:exp}. Quantities $c_1$, $c_2$, and $c_3$ in Theorem~\ref{thm} are the same regardless of the node type, while the values of $\kappa_2$ and $\kappa_3$ depend on the type of node being regressed on the others in \eqref{pl}. 
 We fix $B_1= \kappa_3 \delta_1$ for Bernoulli, Poisson and exponential nodes. For a Gaussian node, this quantity will always equal zero, since $D(\eta_s)=\eta_s^2/2+\log(2\pi)/2$ and hence  $D^{'''}(\eta_s)=0=\kappa_3$. Furthermore, we fix $B_2 = 1/ \delta_3= \delta_2 \kappa_2$ for all four types of nodes. With $B_1$ and $B_2$ fixed, the minimum sample size and the feasible range of the tuning parameter for Bernoulli, Poisson and exponential nodes are exactly the same, 
{as these quantities involve only $B_1$ and $B_2$. 
In particular, from Assumption~\ref{tuning.range}, the range of feasible $\lambda_n$ is $[  8(2-a)\{\log (2p) B_2 /n\}^{1/2} / a, \Lambda_{1}^2 /\{12\Lambda_{2}(d+1) B_1  \log p\}  ],$ 
 and from Assumption~\ref{n}, the minimum sample size is $96^2 (2-a)^2 \Lambda^2_{2} (d+1)^2 B_2 B_1^2 \log (2p) (\log p)^2 / (a^2 \Lambda_{1}^4)$. 
These bounds}
are more restrictive than the corresponding bounds for Gaussian nodes  in Corollary~\ref{col1}.  We now derive a lower bound for the probability of successful neighbourhood recovery for each node type.

\begin{example} \label{ex:gau} 
{If $x_s$ is a Gaussian node, then}
the log-partition function is $D(\eta_s)=\eta_s^2/2+\log(2\pi)/2$. It follows that $D^{''}(\eta_s)=1=\kappa_2$. Thus, $\delta_2=B_2$. 
 By Corollary~\ref{col1}, a lower bound for the probability of successful neighbourhood recovery is
\begin{equation}
 \text{pr}\{\hat{N}(x_s)=N(x_s)\} \geq 1-\exp(-c_2 B_2^2 n)-\exp(-c_3 n /B_2).
 \label{prob:Gaussian}
\end{equation}
\end{example}

\begin{example}\label{ex:bin}
{If $x_s$ is a Bernoulli node, then}
the log-partition function is $D(\eta_s)=\log\{ \exp(-\eta_s) +\exp(\eta_s)\}$, so that $|D^{''}(\eta_s)| \leq 1$ and $|D^{'''}(\eta_s)| \leq 2$. Consequently, $\delta_2=B_2$, and $\delta_1= B_1/\kappa_3=B_1/2$. 
 By Theorem~\ref{thm}, a lower bound for the probability of successful neighbourhood recovery is
\begin{equation}
 \text{pr}\{\hat{N}(x_s)=N(x_s)\} \geq 1-c_1 p^{-B_1/2+2}-\exp(-c_2 B_2^2 n)-\exp(-c_3 n/B_2).
 \label{prob:Bernoulli}
\end{equation}
\end{example}

\begin{example}\label{ex:poi}
{If $x_s$ is a Poisson node, then}
 the log-partition function  is $D(\eta_s)=\exp (\eta_s)$, so $D^{''}(\eta_s)=D^{'''}(\eta_s)= \exp (\eta_s)$. To bound $D^{''}(\eta_s)$ and $D^{'''}(\eta_s)$, we need to bound $\exp(\eta_s)$.  
 Recall from Table~\ref{tab1} that strong compatibility requires that  $\theta_{ts} x_t \leq 0$ when $x_t$ is Gaussian, Poisson or exponential. 
 Therefore, an upper bound for $\exp(\eta_s)$ is
\begin{equation}
\exp ( \eta_s ) \leq \exp\left( \alpha_{1s} + \sum\limits_{t \in I} |\theta_{ts}|\right) \equiv b_P,
\label{b:Poisson}
\end{equation}
with $I$  the set of Bernoulli nodes.  Therefore, $\kappa_2=\kappa_3= b_P$, and so  $\delta_2=B_2/b_P$ and $\delta_1=B_1/b_P$. 
 By Theorem~\ref{thm}, a lower bound on the probability of successful neighbourhood recovery is
\begin{equation}
 \text{pr}\{\hat{N}(x_s)=N(x_s)\} \geq 1-c_1 p^{-  B_1/b_P+2}-\exp(-c_2 B_2^2 n/b_P^2)-\exp(-c_3 n/B_2).
 \label{prob:Poisson}
\end{equation}
\end{example}

\begin{example}\label{ex:exp}  
{If $x_s$ is an exponential node, then}
  the log-partition function  is $D(\eta_s)=-\log (-\eta_s)$, so $D^{''}(\eta_s)=\eta_s^{-2}$ and $D^{'''}(\eta_s)=-2\eta_s^{-3}$. Furthermore,
  \begin{equation} \label{newnum}
  \eta_s = \alpha_{1s} + \sum_{t \neq s} \theta_{ts} x_t  \leq  \alpha_{1s} + \sum_{t \in I} \theta_{ts} x_t { \leq  \alpha_{1s} + \sum_{t \in I} |\theta_{ts}| < 0},
  \end{equation}
  with $I$  the set of Bernoulli nodes. In \eqref{newnum},
 the first inequality follows from  the requirement  for compatibility from Table~\ref{tab1} that $\theta_{ts} x_t \leq 0$ when $x_t$ is Gaussian, Poisson or exponential; the second inequality follows from the fact that Bernoulli nodes are coded as 
 $+1$ and $-1$; and the third inequality follows from the Bernoulli-exponential entry in 
   Table~\ref{tab1}. Therefore, it follows that
     \begin{equation}
  |\eta_s| { \geq \left|\alpha_{1s} + \sum_{t \in I}  |\theta_{ts}| \right| } \geq  |\alpha_{1s}| - \sum_{t \in I} |\theta_{ts}|  \equiv b_E. \label{b:exp}
  \end{equation}
  As a result, $|D^{''}(\eta_s)|$ and $|D^{'''}(\eta_s)|$ are bounded by $\kappa_2=b_E^{-2}$ and  $\kappa_3=2b_E^{-3}$, respectively. For fixed $B_1$ and $B_2$, we have $\delta_2=b_E^2 B_2$ and $\delta_1=B_1 b_E^3/2$. 
 By Theorem~\ref{thm}, a lower bound for the probability of successful neighbourhood recovery is
\begin{equation}
 \text{pr}\{\hat{N}(x_s)=N(x_s)\} \geq 1-c_1 p^{- b_E^3 B_1/2+2}-\exp(-c_2 b_E^4 B_2^2 n)-\exp(-c_3 n/B_2) .
 \label{prob:exp}
\end{equation}
\end{example}

Examples~\ref{ex:gau}-\ref{ex:exp} reveal that the neighbourhood of a Gaussian node is easier to recover than the neighbourhood of the other three types of nodes: the first requires a smaller minimum sample size when $p$ is large, allows for a wider range of feasible tuning parameters, and has in general a higher probability of success. 
 As a result, the neighbourhood of the Gaussian node should be used when estimating an edge between a Gaussian node and a non-Gaussian node.

Which neighbourhood should we use to estimate an edge between two non-Gaussian nodes? There are no clear winners: while \eqref{prob:Bernoulli} can be evaluated given knowledge of $c_1$, $c_2$, and $c_3$, \eqref{prob:Poisson} and \eqref{prob:exp} also require knowledge of the unknown quantities $b_E$ and $b_P$, which are functions of unknown quantities $\Theta_s$ and $\alpha_{1s}$ in \eqref{b:Poisson} and \eqref{b:exp}.  
One possibility is to insert a consistent estimator for these parameters (see e.g. \citealt{vandeGeer2008}, \citealt{bunea2008}) in order to obtain a consistent estimator for $b_P$ or $b_E$. This leads to the following lemma.

\begin{lemma}
Suppose $\tilde{\Theta}_s$ and $\tilde{\alpha}_{1s}$ are consistent estimators of the true parameters in the conditional densities \eqref{cond:Poisson} and \eqref{cond:exp}. Let $I$ be the index set of the Bernoulli nodes. 

1. If $x_s$ is a Poisson node and $\tilde{b}_P= \exp( \tilde{\alpha}_{1s} + \sum_{t \in I} |\tilde{\theta}_{ts}|)$, then  
\begin{equation}
1-c_1 p^{-  B_1/\tilde{b}_P+2}-\exp(-c_2  B_2^2 n/\tilde{b}^2_P)-\exp(-c_3 n/B_2)
\label{estprob:Poisson}
\end{equation}
is a consistent estimator of a lower bound for $\text{pr}\{ \hat{N}(x_s) = N(x_s)\}$.

2. If $x_s$ is an exponential node and $\tilde{b}_E= |\tilde{\alpha}_{1s}| - \sum_{t \in I} |\tilde{\theta}_{ts}|$, then
\begin{equation}
1-c_1p^{- \tilde{b}_E^3 B_1/2+2}-\exp(-c_2 \tilde{b}_E^4 B_2^2 n)-\exp(-c_3 n/B_2)
\label{estprob:exp}
\end{equation}
 is a consistent estimator of a lower bound for $\text{pr}\{ \hat{N}(x_s) = N(x_s)\}$.
\label{lmm:selection}
\end{lemma}

Therefore, by inserting consistent estimators of $\Theta_s$ and $\alpha_{1s}$ into \eqref{b:Poisson} or \eqref{b:exp}, we can reconstruct an edge by choosing the estimate with the highest probability of correct recovery according to 
 \eqref{prob:Bernoulli}, \eqref{estprob:Poisson}, and \eqref{estprob:exp}. The rules are summarized in Table~\ref{tab2}. The results in this section illustrate a worst case scenario for recovery of each neighbourhood, in that Theorem \ref{thm}  provides a lower bound for the probability of successful neighbourhood recovery.

\begin{table}
\def~{\hphantom{0}}
\caption{Neighbourhood to use in estimating an edge between two non-Gaussian nodes of different types} 
\begin{center}
\begin{tabular}{ p{3.5cm}p{10cm} }
 & Selection rules\\
Poisson \& Exponential  & Choose Poisson if $\tilde{b}_E^2 \tilde{b}_P < 1$ and  $\tilde{b}_E^3 \tilde{b}_P < 2$. Choose exponential if
$\tilde{b}_E^2 \tilde{b}_P > 1$ and  $\tilde{b}_E^3 \tilde{b}_P > 2$.\\
Poisson \& Bernoulli &  Choose Poisson if $\tilde{b}_P < 1$. Choose Bernoulli if $\tilde{b}_P > 2$. 
\\
Exponential \& Bernoulli & Choose exponential if $\tilde{b}_E \geq 1$. Choose Bernoulli if $\tilde{b}_E <1 $. \\
\end{tabular}
\end{center}
\label{tab2}
When the conditions in this table are not met, there is no clear preference in terms of which neighbourhood to use.
\end{table}

\section{Neighbourhood Recovery and Selection with Partially-Specified Models}\label{nojd}

In Section~\ref{theory}, we showed that the neighbourhood selection approach of Section~\ref{estimation} can recover the true graph when each node's conditional distribution is of the form \eqref{cond}, provided that the conditions for strong compatibility are satisfied.  In this section, we consider a partially-specified model in which some of the nodes are assumed to have conditional distributions of the form \eqref{cond}, and we make no assumption on the conditional distributions of the remaining nodes.  We will show that in this setting, neighbourhoods of the nodes with conditional distributions of the form (3) can still be recovered.

Here the neighbourhood of $x_s$ is defined based upon its conditional density, \eqref{cond}, as $N^{0}(x_s)=\{t: \theta_{ts}\neq 0\}$. Assumption~\ref{A} in Section~\ref{subsec:jd} is inappropriate since we no longer assume that all $p$ nodes have conditional densities of the form \eqref{cond}, and consequently we are not assuming a particular form for the joint density. Therefore, we make the following assumption to replace Propositions~\ref{prop.e1} and \ref{prop.e2}.
\begin{assumption} Assume that (i) $\text{pr}(\xi_1)\geq  1- c_1  p^{-\delta_1+2},$ (ii) $\text{pr}(\xi_2)\geq 1-\exp(-c_2 \delta_2^2 n ).$
\label{A2}
\end{assumption}
 \begin{theorem}
Suppose that the $s$th node has conditional density \eqref{cond}, and that Assumptions~\ref{irrep}\ -- \ref{D} and \ref{thetamin} -- \ref{A2} hold. Then with probability at least $1-c_1 p^{-\delta_1+2}-\exp(-c_2 \delta_2^2 n)-\exp(-c_3 \delta_3 n) $, for some constants $c_1, c_2, c_3$, and $\delta_3=1/(\kappa_2 \delta_2)$, the {estimator from} \eqref{pl} recovers the true neighbourhood of $x_s$ exactly, so that $\hat{N}(x_s)=N^{0}(x_s)$.
 \label{thm2}
\end{theorem}

The proof of Theorem~\ref{thm2} is similar to that of Theorem~\ref{thm}, and is thus omitted.  Theorem~\ref{thm2} indicates that our neighbourhood selection approach can recover the neighbourhood of any node for which the conditional density is of the form  \eqref{cond}, provided that Assumption~\ref{A2} holds.  This means that in order to recover an edge between two nodes using our neighbourhood selection approach, it suffices for one of the two nodes' conditional densities to be of the form \eqref{cond}. Consequently, we can model relationships that are far more flexible than those outlined in Table~\ref{tab1}, e.g. an edge between a Poisson node and a node  that takes values on the whole real line.

Although Theorem~\ref{thm2} allows us to go beyond some of the restrictions in Table~\ref{tab1}, it is still restricted in that it only guarantees recovery of an edge between two nodes for which at least one of the  node-conditional densities is exactly of the form \eqref{cond}. In future work, we could  generalize Theorem~\ref{thm2} to the case where  \eqref{cond}  is simply an approximation to  the true node-conditional distribution.

\section{Numerical Studies}\label{simulation}
\subsection{Data Generation}\label{generate}

We consider mixed graphical models with two types of nodes, and $m=p/2$ nodes per type, for Gaussian-Bernoulli and Poisson-Bernoulli models. We order the nodes so that the Gaussian or Poisson nodes precede the Bernoulli nodes. 

For both models, we construct a graph in which the $j$th node for $j=1,\ldots,m$ is connected with the adjacent nodes of the same type, as well as the $(m+j)$th node of the other type, as shown in Fig.~\ref{NEWFIG}. This encodes the edge set $E$. 
For $(i,j) \in E$ and $i<j$, we generate the edge potentials $\theta_{ij}$ and $\theta_{ji}$ as
\begin{equation}
 \theta_{ij} = \theta_{ji}= y_{ij}  r_{ij}, \; \text{pr}(y_{ij}= 1) =\text{pr}(y_{ij}= -1) =0.5, \; \ r_{ij} \sim \text{Unif}(a,b).
\label{generate.theta}
 \end{equation} 
We set  $\theta_{ij} = \theta_{ji}=0$ if $(i,j) \notin E$. 
 Section~\ref{generation_detail} in the Supplementary Material lists additional steps to ensure strong compatibility of the conditional distributions. 
Values of $a$ and $b$ in \eqref{generate.theta}, as well as the parameters of $f_s(x_s)$ in the conditional density \eqref{cond}, are specified in Sections~\ref{probrecovery}--\ref{PB}.

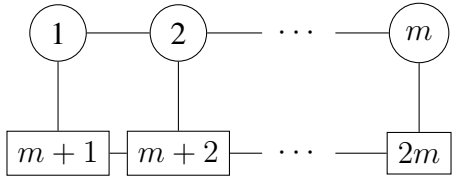
\begin{figure}[t]
\begin{center}
\begin{tikzpicture}
  [scale=.8,auto=left,  every node/.style={circle,fill=blue!20}]

  \node[draw=black, fill=white,circle]  (g1) at (1,5) {1};
  \node[draw=black, fill=white,circle]  (g2) at (3,5)  {2};
  \node[draw=none,fill=none] (g3) at (5,5) {$\cdots$};
  \node[draw=black, fill=white,circle]  (g4) at (7,5)  {$m$};
  \node[draw=black, fill=white,rectangle] (b1) at (1,3) {$m+1$};
  \node[draw=black, fill=white,rectangle] (b2) at (3,3)  {$m+2$};
 \node[draw=none,fill=none] (b3) at (5,3) {$\cdots$}; 
  \node[draw=black, fill=white,rectangle] (b4) at (7,3)  {$2m$};
  \foreach \from/\to in {g1/g2,g2/g3,g3/g4,g1/b1,g2/b2, g4/b4, b1/b2, b2/b3, b3/b4}
    \draw (\from) -- (\to);
\end{tikzpicture}
\end{center}
\caption{The graph used to generate the data in Section~\ref{simulation}. There are $m=p/2$ Gaussian or Poisson nodes, shown as circles, and $m=p/2$ Bernoulli nodes, shown as rectangles. }
\label{NEWFIG}
\end{figure}

To sample from the joint density $p( {x})$ in \eqref{joint} without calculating the log-partition function $A$, we employ a Gibbs sampler, as in \citet{lee2012}. Briefly, we iterate through the nodes, and sample from each node's conditional distribution. To ensure independence, after a burn-in period of 3000 iterations, we select samples from the chain 500 iterations apart from each other.

\subsection{Probability of Successful Neighbourhood Recovery}\label{probrecovery}

In Section~\ref{subsec:jd} we saw that the probability of successful neighbourhood recovery for  neighbourhood selection  converges to 1 exponentially fast with the sample size. And in Section~\ref{selection} we saw that the estimates from the Gaussian nodes are superior to those from the Bernoulli nodes, in the sense that a smaller sample size is needed in order to achieve a given probability of successful recovery. We now verify those findings empirically.  Here, successful neighbourhood recovery is defined to mean that the estimated and true edge sets of a graph or a sub-graph are identical.

We set $a=b=0\cdot$3 in \eqref{generate.theta} so that Assumption~\ref{thetamin} is satisfied, and generate one Gaussian-Bernoulli graph for each of $p=60$, $p=120$, and $p=240$. We set $\alpha_{1s}=0$ and $\alpha_{2s}= -1$ in \eqref{cond:Gaussian} for Gaussian nodes, and $\alpha_{1s}=0$ for Bernoulli nodes \eqref{cond:Bernoulli}.  For each graph, $100$ independent data sets are drawn from the Gibbs sampler. We perform neighbourhood selection using the estimator from \eqref{pl.weighted}, with the tuning parameter $\lambda_n$  set to be a constant $c$ times $\{\log (p)/n\}^{1/2}$, so that it is on the scale required by Assumption~\ref{tuning.range}, as illustrated in Remark~\ref{remark3}.

In order to achieve successful neighbourhood recovery as the sample size increases, the value of $c$ must be in a range matching the requirement of Assumption~\ref{tuning.range}. {We explored a range of values of $c$, and in Fig.~\ref{fig1} we show the probability of successful neighbourhood recovery for $c=2.6$.}
For ease of viewing, we display separate empirical probability curves for the Gaussian-Gaussian, Bernoulli-Bernoulli, and Bernoulli-Gaussian subgraphs. Panels (a) and (b) are estimates obtained by regressing the Gaussian nodes onto the others, and panels (c) and (d) are the estimates from regressing the Bernoulli nodes onto the others. 
We see that the probability of successful recovery  increases to 1 once the scaled sample size exceeds the threshold required in Assumption~\ref{n} and Corollary~\ref{col1}. Furthermore, {panels (b) and (c)} agree with the conclusions of Section~\ref{selection}:   neighbourhood recovery using the regression of a Gaussian node onto the others requires fewer samples than recovery using the regression of a Bernoulli node onto the others.

\begin{figure}[t]
   \centering
   \subfigure{\includegraphics[scale=0.27]{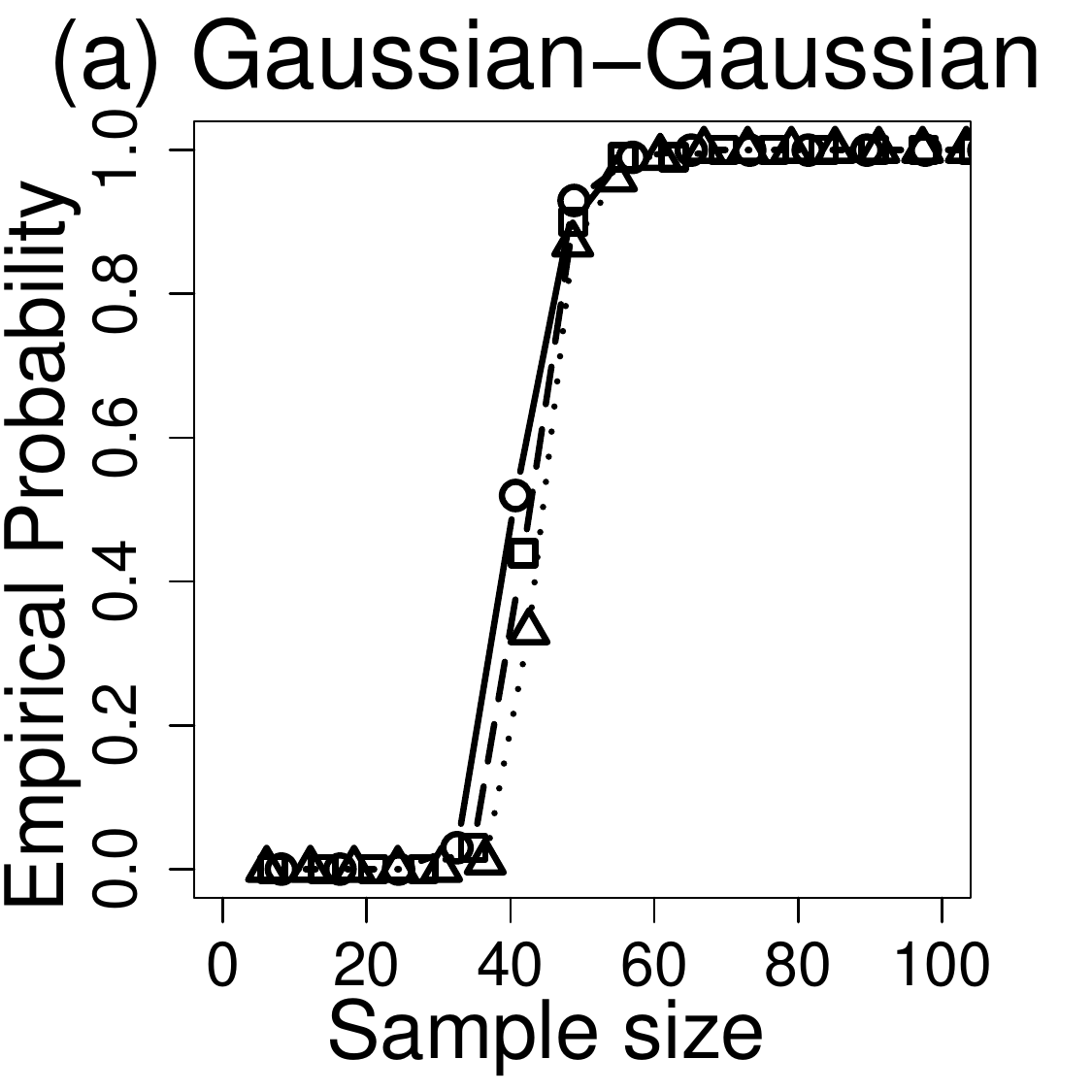}}
   \subfigure{\includegraphics[scale=0.27]{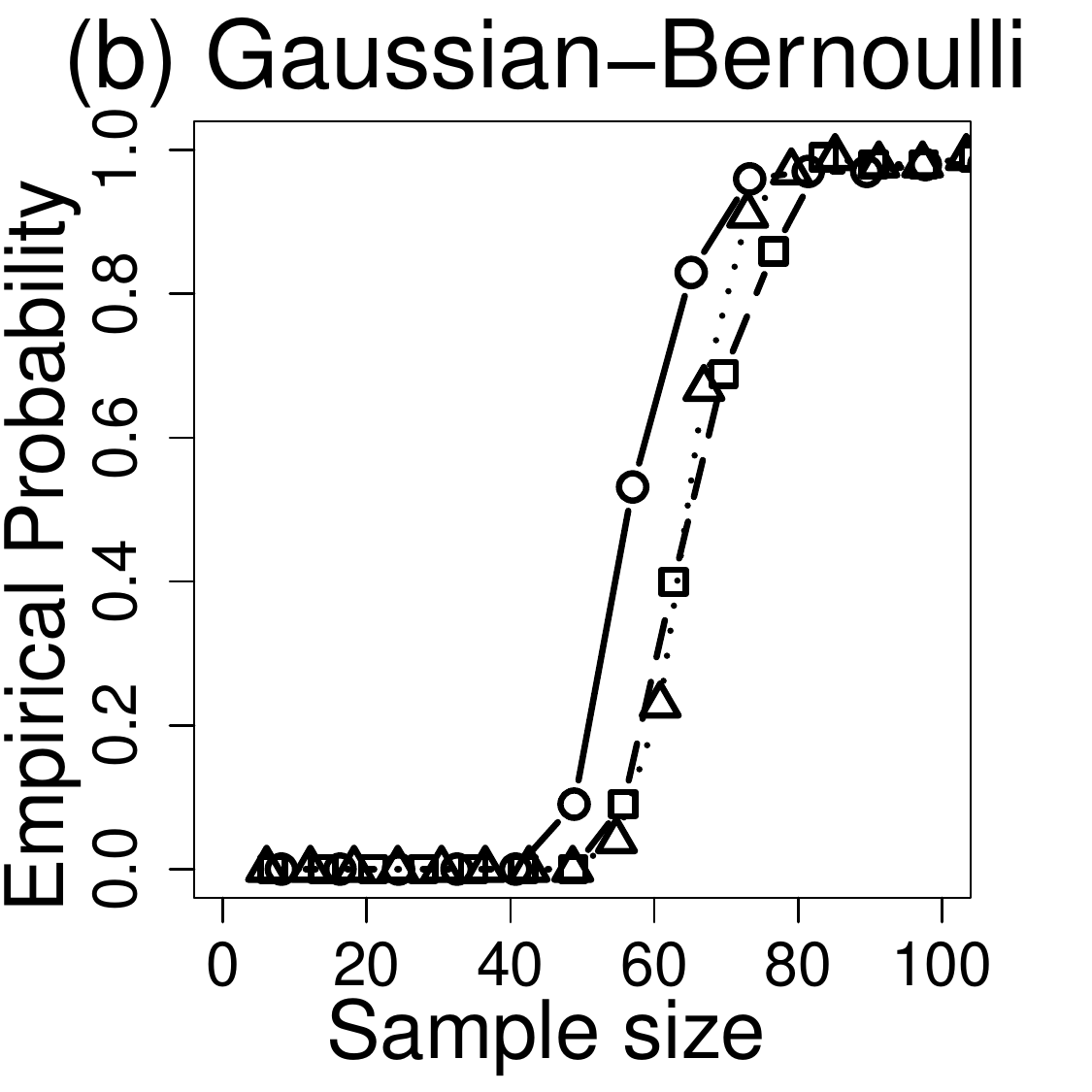}}
   \subfigure{\includegraphics[scale=0.27]{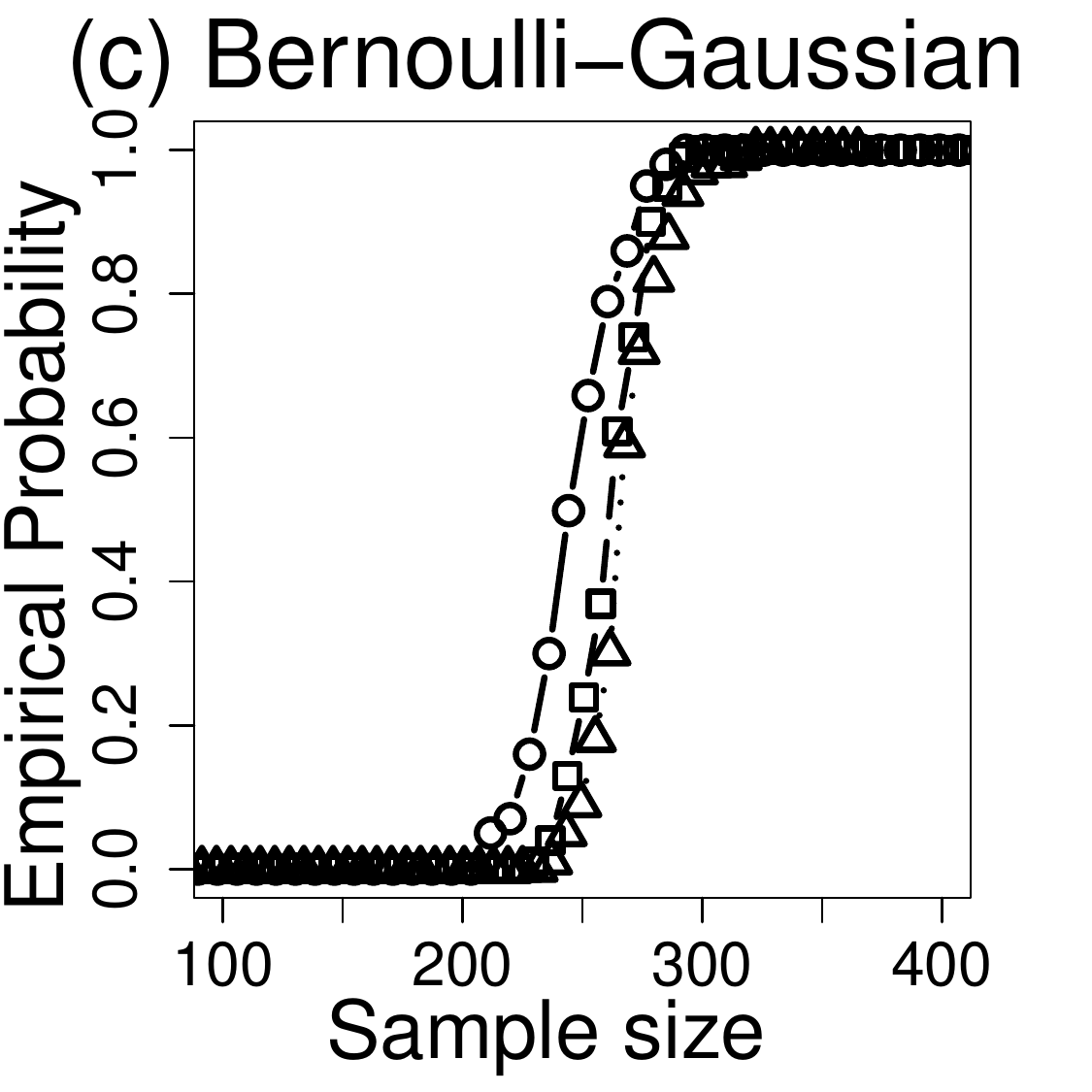}}
   \subfigure{\includegraphics[scale=0.27]{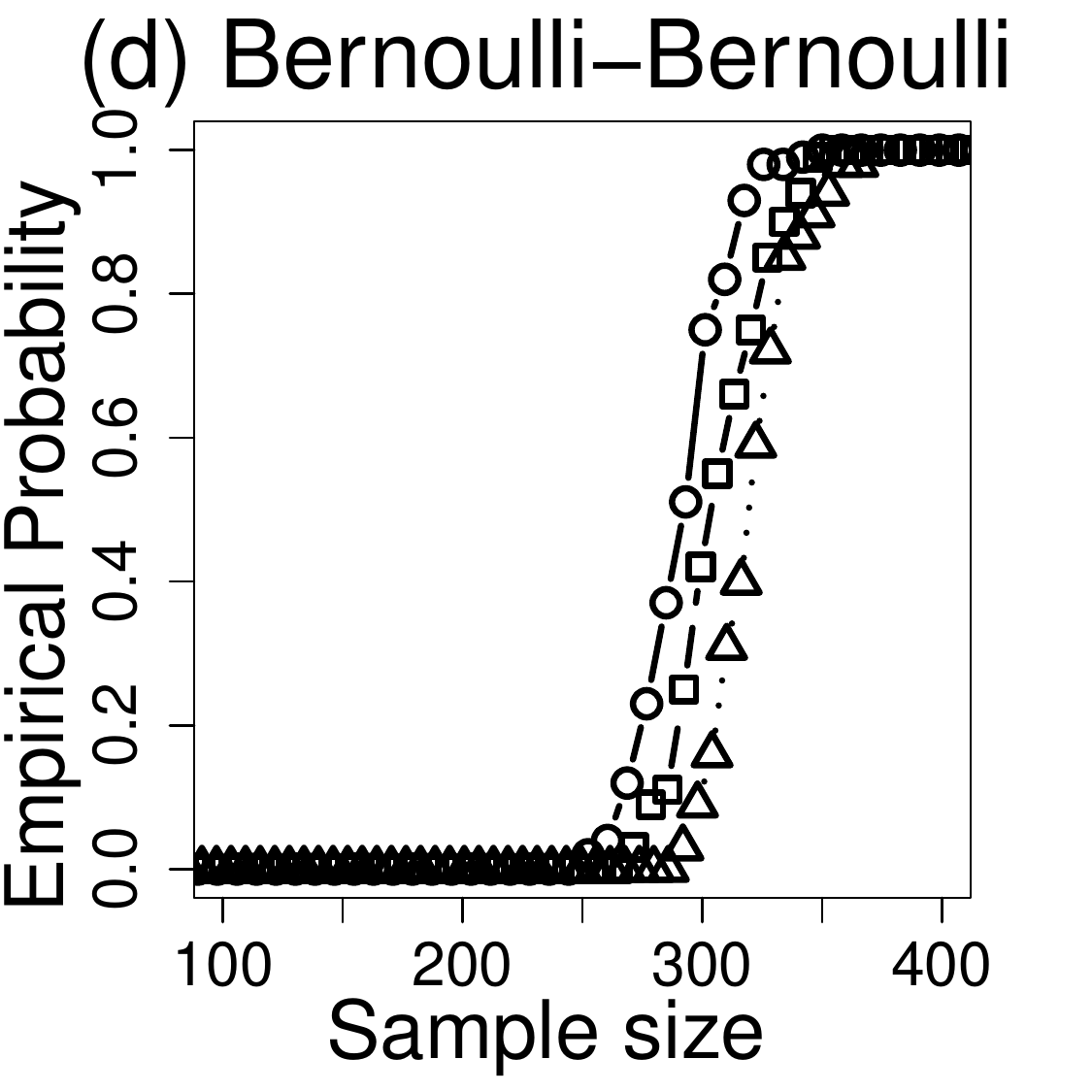}}
\caption{Probability of successful neighbourhood recovery, $y$-axis, as a function of  scaled sample size $n/\{3 \log(p)\}$, $x$-axis, for the set-up of Section~\ref{probrecovery}. The  curves are empirical probabilities of successful neighbourhood recovery for graphs with $60$  (\protect\includegraphics[height=0.5em]{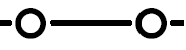}), $120$  (\protect\includegraphics[height=0.5em]{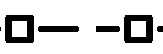}), and $240$ nodes (\protect\includegraphics[height=0.5em]{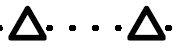}),  averaged over $100$ independent data sets.  The tuning parameter is set to be 2$\cdot$6$\{\log (p) /n\}^{1/2}$. The title of each panel indicates the subgraph for which the recovery probability is displayed, and the first word in the title indicates the node type that was regressed in order to obtain the subgraph estimate. For instance, panel (b) displays probability curves for edges between Gaussian and Bernoulli nodes that are estimated from the $\ell_1$-penalized linear regression of Gaussian nodes. Panel (c) displays the same quantity, estimated via an $\ell_1$-penalized logistic regression of the Bernoulli nodes.}
\label{fig1}
 \end{figure}

\subsection{Comparison to Competing Approaches}\label{TP}

In this section, we compare the proposed method to alternative approaches on a Gaussian-Bernoulli graph. We limit the number of nodes to $p=40$ in order to facilitate comparison with the computationally intensive approach of \citet{lee2012}.  We generate $100$ random graphs  with $a=$0$\cdot$3 and $b=$0$\cdot$6 in (\ref{generate.theta}), and we set $\alpha_{1s}=0$ and $\alpha_{2s}= -1$ in \eqref{cond:Gaussian} for Gaussian nodes and $\alpha_{1s}=0$ for Bernoulli nodes \eqref{cond:Bernoulli}. Twenty independent samples of $n=200$ observations are generated from each graph. We evaluate the performance of each approach by computing  the number of correctly estimated edges as a function of the number of estimated edges in the graph. Results are averaged over 20 data sets from each of 100 random graphs, for a total of 2000 simulated data sets.

Seven approaches are compared in this study: 1) our proposal for neighbourhood selection in the mixed graphical model;  2) penalized maximum likelihood estimation in the mixed graphical model \citep{lee2013}; 3)  weighted $\ell_1$-penalized regression in the mixed graphical model, as proposed by \citet{cheng2013}; 4) graphical random forests \citep{fellinghauer2013};  5) neighbourhood selection in the Gaussian graphical model \citep{meinshausen2006}, where we use an $\ell_1$-penalized linear regression to estimate the neighbourhood of all nodes; 6) the graphical lasso \citep{friedman2008}, which treats all features as Gaussian; and 7) neighbourhood selection in the Ising model \citep{ravikumar2010}, where we use $\ell_1$-penalized logistic regression on all nodes after dichotomizing the Gaussian nodes by their means. The first four methods are designed for mixed graphical models, with \citet{lee2012} and \citet{cheng2013} specifically proposed for Gaussian-Bernoulli networks. In contrast, the last three methods ignore the presence of mixed node types. For methods based on neighbourhood selection, we use the union rule of \citet{meinshausen2006} to reconstruct the edge set from the estimated neighbourhoods, with one exception: to estimate the Gaussian-Bernoulli edges for our proposed method, we use the estimates  from the Gaussian nodes, as suggested by the theory developed in Section~\ref{selection}. 

Due to its high computational cost, the method of \citet{lee2012} is run on 250 data sets from 50 graphs rather than 2000 data sets from 100 graphs. 

\begin{figure}[t]
   \centering
   \subfigure{\includegraphics[scale=0.22 ]{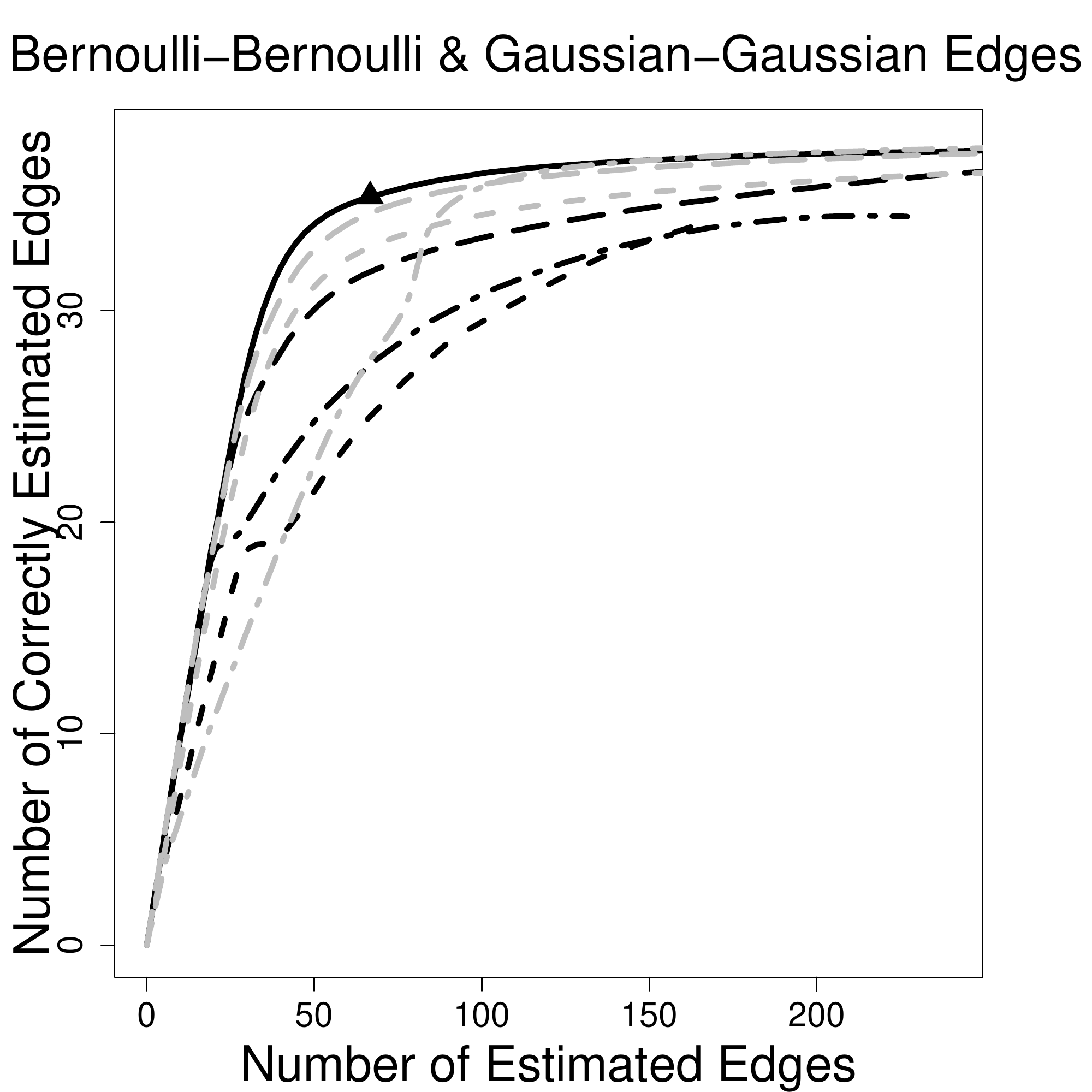}}\qquad
   \subfigure{\includegraphics[scale=0.22]{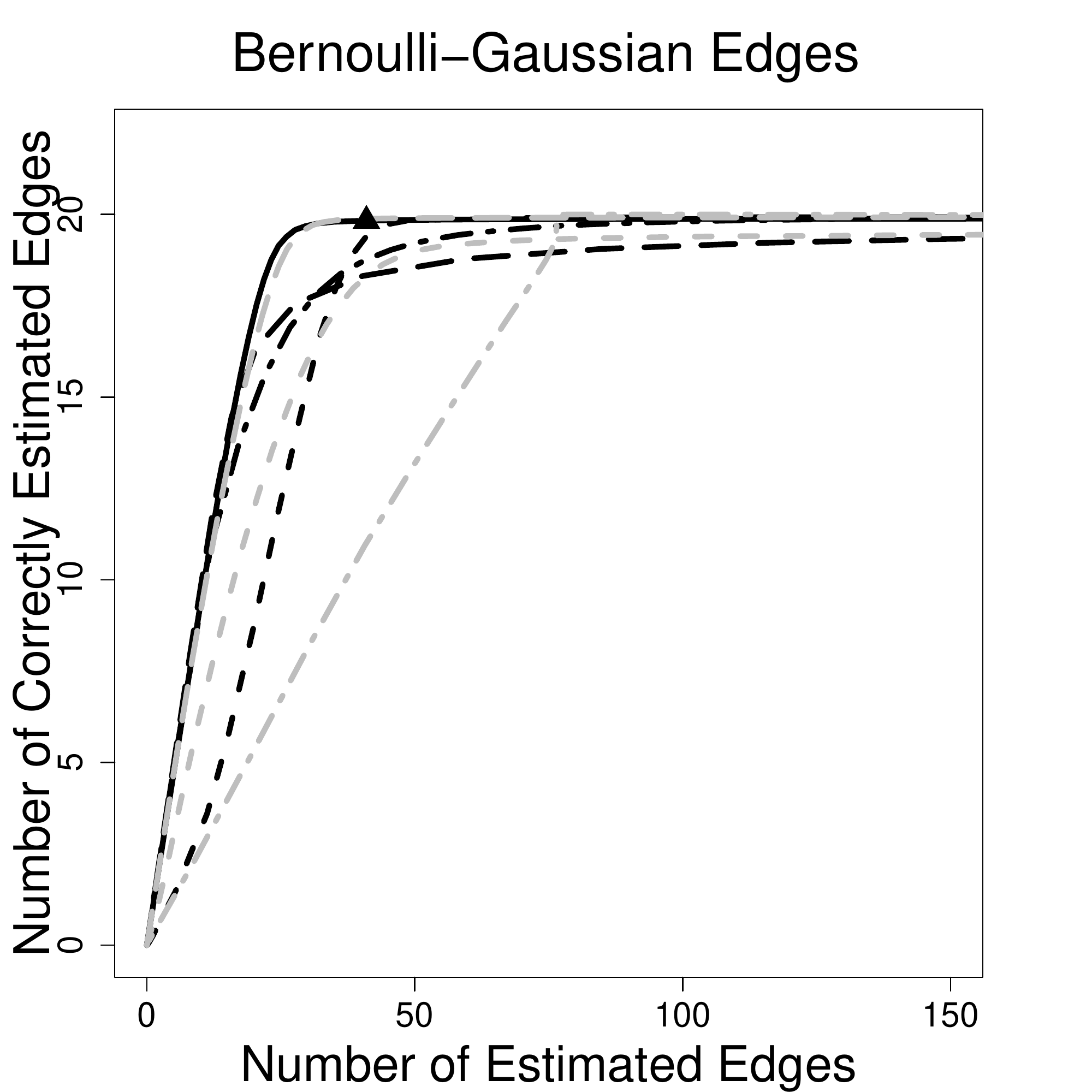}}\\
\caption{ {Simulation results for the Gaussian-Bernoulli graph, as described in Section~\ref{TP}. The number of correctly estimated edges is displayed as a
function of the number of estimated edges, for a range of tuning parameter values in a graph with $p=40$ and $n=200$. The left panel corresponds to edges between nodes of the same type, while the right panel corresponds to the edges between Gaussian and Bernoulli nodes. 
The curves within each panel represent our proposal (\protect\includegraphics[height=0.5em]{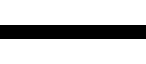}), \citet{lee2012} (\protect\includegraphics[height=0.5em]{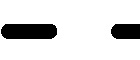}), \citet{cheng2013} (\protect\includegraphics[height=0.5em]{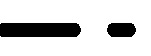}),  \citet{fellinghauer2013} (\protect\includegraphics[height=0.5em]{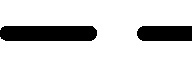}), neighbourhood selection in the Gaussian graphical model  (\protect\includegraphics[height=0.5em]{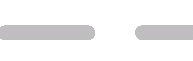}), neighbourhood selection in the Ising model (\protect\includegraphics[height=0.5em]{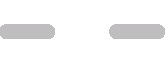}), and the graphical lasso (\protect\includegraphics[height=0.5em]{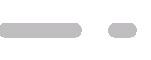}). The black triangle shows the average performance of our proposed approach with the tuning parameter selected by the Bayesian information criterion (Section~\ref{tuning}).} }
\label{fig2}
\end{figure}

The left-hand panel  of Fig.~\ref{fig2} displays results for Bernoulli-Bernoulli and Gaussian-Gaussian edges, and the right-hand panel displays  results for edges between Gaussian and Bernoulli nodes.

The curves in Fig.~\ref{fig2} correspond to the estimated graphs as the tuning parameter for each method is varied. Recall from Section~\ref{tuning} that our proposal involves a tuning parameter $\lambda_n^{G}$ for the $\ell_1$-penalized linear regressions of the Gaussian nodes onto the others, and a tuning parameter $\lambda_n^{B}$ for the $\ell_1$-penalized logistic regressions of the Bernoulli nodes onto the others. 
The triangles in Fig.~\ref{fig2} show the average performance of  our proposed method with the tuning parameters $\hat{\lambda}_n^{B}$ and $\hat{\lambda}_n^{G}$ selected using  \textsc{bic} summed over the Bernoulli and Gaussian nodes, respectively, as described in Section~\ref{tuning}. This choice yields good precision ($52\%$) and recall  ($95\%$) for edge recovery in the graph. To obtain the curves in Fig.~\ref{fig2}, we set $\lambda_n^{B} = (\hat{\lambda}_n^{B}/\hat{\lambda}_n^{G}) \lambda_n^{G}$, and varied the value of $\lambda_n^G$.

In general, our proposal outperforms the competitors, which is expected since it assumes the correct model. Though the proposals of \citet{lee2012} and \citet{cheng2013}  are intended for a Gaussian-Bernoulli graph, they attempt to capture more complicated relationships than in \eqref{joint}, and so they perform worse than our proposal. On the other hand, the graphical random forest of \citet{fellinghauer2013} performs reasonably well, despite the fact that it is a nonparametric approach. Neighbourhood selection in the Gaussian graphical model performs closest to the proposed method in terms of edge selection. 
The Ising model suffers substantially due to dichotomization of the Gaussian variables.  The graphical lasso algorithm experiences serious violations to its multivariate Gaussian assumption, leading to poor performance.

\subsection{Application of Selection Rules for Mixed Graphical Models}\label{PB}

In Section~\ref{TP}, in keeping with the results of Section~\ref{selection}, we always used the estimates from the Gaussian nodes in estimating an edge between a Bernoulli node and a Gaussian node.
 Here we consider a mixed graphical model of Poisson and Bernoulli nodes. In this case, the selection rules in Section~\ref{selection} are more complex, and whether it is better to use a Poisson node or a Bernoulli node in order to estimate a Bernoulli-Poisson edge depends on the true parameter values in Table~\ref{tab2}. 

We generate a graph with $p=80$ nodes as follows: $a=0.8$ and $b=1$ in (\ref{generate.theta}), $\alpha_{1s}=-3$ for $s=1, \ldots, 20$ and $\alpha_{1s}=0$ for $s=21, \ldots, 40$ for the Poisson nodes, and $\alpha_{1s}=0$ for the Bernoulli nodes.  This guarantees that $b_P$ in \eqref{b:Poisson} is smaller than 1 for the first half of the Poisson nodes, and larger than 2 for the second half, due to the structure of the graph from Fig.~\ref{NEWFIG}. In order to estimate a Bernoulli-Poisson edge, we will use the estimates from the Poisson nodes if $b_P<1$ and  the estimates from the Bernoulli nodes if $b_P>2$, according to the selection rules in Table~\ref{tab2}.

We compare the performance of our proposed approach using the selection rules in Table~\ref{tab2}, with the true and estimated parameters, to   our proposed approach using the union and intersection rules (Section~\ref{selection}), as well as the graphical random forest of \citet{fellinghauer2013}. To prevent over-shrinkage of the parameters for estimation of $b_P$ in \eqref{b:Poisson}, we set $\lambda_n$ in \eqref{pl} to equal 0$\cdot$5 times the value from the Bayesian information criterion for each node type. We present only the results for Poisson-Bernoulli edges, as the selection rules in Section~\ref{selection} apply to edges between nodes of different types.

\begin{figure}[t]
   \centering
   \subfigure{\includegraphics[scale=0.25]{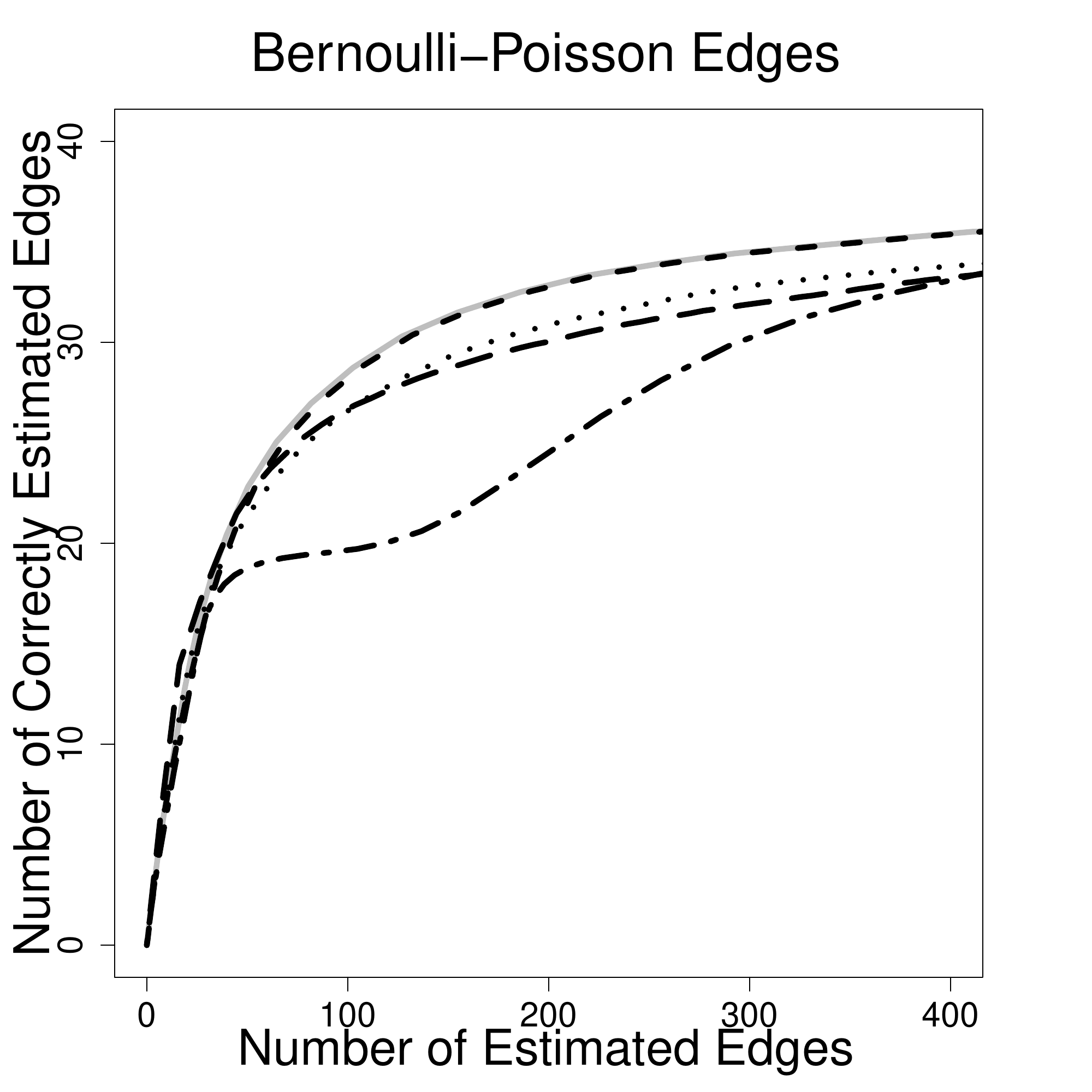}}
\caption{Summary of the simulation results for the Poisson-Bernoulli graph, as described in Section~\ref{PB}. The number of correctly estimated edges is displayed as a
function of the number of estimated edges, for a range of tuning parameter values in a graph with $p=80$ nodes from $n=200$ observations. 
The curves represent the selection rule from Section~\ref{selection} with the true parameters (\protect\includegraphics[height=0.5em]{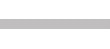}), the selection rule from Section~\ref{selection} with estimated parameters  (\protect\includegraphics[height=0.5em]{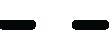}), the union rule (\protect\includegraphics[height=0.5em]{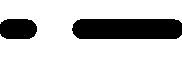}), the intersection rule (\protect\includegraphics[height=0.5em]{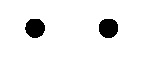}), and the method from \citet{fellinghauer2013} (\protect\includegraphics[height=0.5em]{GRaFo.jpg}).}
\label{fig3}
\end{figure}

Results are shown in Fig.~\ref{fig3}, averaged over 20 samples from each of 25 random graphs. The selection rules proposed in Section~\ref{selection} clearly outperform the commonly-used union and intersection rules. 
The curve for the selection rule from Section~\ref{selection} using the estimated parameter values is almost identical to the curve using the true parameter values, which indicates that in this case the quantity  $b_P$ is accurately estimated for each node.   The graphical random forest slightly outperforms our proposal when few edges are estimated, but performs worse when the estimated graph includes more edges. This may indicate that as the graph becomes less sparse, the nonparametric graphical random forest approach suffers from insufficient sample size.

\section{Discussion}\label{discussion}

In Section~\ref{compatibility} we saw that a stringent set of restrictions is required for compatibility or strong compatibility of the node-conditional distributions given in \eqref{cond:Gaussian}--\eqref{cond:exp}. These restrictions limit the theoretical flexibility of the conditionally-specified mixed graphical model, especially when modeling unbounded variables. It is possible that by truncating unbounded variables, we may be able to circumvent some of these restrictions. Furthermore, the  model \eqref{joint} assumes pairwise interactions in the form of $x_s x_t$, which can be seen as a second-order approximation of the true edge potentials in \eqref{PGM}. We can relax this assumption by fitting non-linear edge potentials using semi-parametric penalized regressions, as in \citet{voorman2014}.

{
\texttt{R} code for replicating the numerical results in this paper is available at \url{https://github.com/ChenShizhe/MixedGraphicalModels}. 
}

\section*{Acknowledgement}
{We thank Jie Cheng, Bernd Fellinghauer, and Jason Lee for providing code and responding to our inquiries.}  
This work was partially supported by National Science Foundation grants to D.W. and A.S., National Institutes of Health grants to D.W. and A.S., and a Sloan Fellowship to D.W.

\bibliographystyle{chicago}
\bibliography{paper-ref}
\appendix

\newpage
\section{Appendix}

\subsection{A Proof for Proposition~\ref{compat:prop}}\label{proof:compat}
\begin{proof}
First of all, it is easy to see that if $\theta_{st}=\theta_{ts}$, then any function $g$ such that
\begin{equation}
g(\utilde{x}) \propto \exp \left\{ \sum\limits_{s=1}^p f_s(x_s) + \frac{1}{2} \sum\limits_{s=1}^{p}\sum\limits_{t \neq s} \theta_{ts} x_s x_t\right\}
\label{compat:proof}
\end{equation}
is capable of generating the conditional densities in \eqref{cond}   as long as the function $g$ is integrable with respect to $x_s$ for $s = 1,\ldots, p$. The function $g$ can be decomposed as
\begin{equation*}
g(\utilde{x}) \propto \exp \left\{  f_s(x_s) + \frac{1}{2}  \sum\limits_{t: t \neq s}{(\theta_{ts}+\theta_{st})} x_s x_t\right\} \exp\left\{\sum\limits_{t\neq s}f_t(x_t)+
 \frac{1}{2}  \sum\limits_{t: t\neq s, \ j: j\neq s,\  j\neq  t} \theta_{tj} x_j x_t \right\},
\end{equation*}
so the integrability of the conditional density $p(x_s\mid x_{-s})$ guarantees the integrability of $g$ with respect to $x_s$. Therefore, the conditional densities of the form in \eqref{cond} are compatible if $\theta_{ts}=\theta_{st}$.

We now prove that any function $h$ that is capable of generating the conditional density in \eqref{cond}   is in the form \eqref{compat:proof}. The following proof is essentially the same as that in \citet{besag1974}. Suppose $h$ is a function that is capable of generating the conditional densities. Define 
$P(\utilde{x})=\log \{ h(\utilde{x})/h(\utilde{0})\}$, where $\utilde{0}$ can be replaced by any interior point in the sample space. 

By definition, $P(\utilde{0})=\log\{h(0)/h(0) \}=0$. Therefore,  $P$ can be written in the general form
$$ P(\utilde{x})= \sum\limits_{s=1}^p x_s G_s (x_s) + \sum\limits_{t\neq s}   \frac{G_{ts} (x_t,x_s)}{2} x_t x_s + \sum\limits_{t \neq s,t\neq j,j\neq s}  \frac{G_{tsj} (x_t,x_s,x_j)}{6} x_t x_s x_j + \cdots \ , $$
where we write the function $P$ as the sum of interactions of different orders. Note that the factor of $1/2$ is due to $G_{st}(x_s,x_t)=G_{ts}(x_t,x_s)$; similar factors apply for higher-order interactions. 
Recalling that we assume $h$ is capable of generating the conditional density $p(x_s\mid x_{-s})$, from Definition~\ref{compat:defn}   we know that 
$$P(\utilde{x})-P(\utilde{x}_{s}^0)= {\log \left\{ \frac{h(x)/\int h(x) dx_s }{h(x_s^0)/\int h(x) dx_s } \right\}} = \log \left\{ \frac{p(x_s\mid \utilde{x}_{-s})}{p(0 \mid \utilde{x}_{-s}) } \right\}, $$
where $\utilde{x}_{s}^0={(x_1,\ldots, x_{s-1}, 0, x_{s+1},\ldots, x_{p})}^\T$ and $p(x_s\mid x_{-s})$ is the conditional density  in \eqref{cond}. It follows that 
 \begin{equation}
 \log \left\{ \frac{p(x_s\mid \utilde{x}_{-s})}{p(0 \mid \utilde{x}_{-s}) } \right\} = P(\utilde{x})-P(\utilde{x}_{s}^0)= x_s \left(G_s (x_s) + \sum\limits_{t: t\neq s} x_t G_{ts} (x_t,x_s)+\cdots \right).
\label{Q:xs}
\end{equation}

Letting $x_t=0$ for $t \neq s$ in \eqref{Q:xs} and using the form of the conditional densities in \eqref{cond}  , we have
\begin{equation}
x_s G_s (x_s) = f_s(x_s)-f_s(0).
\label{first.order}
\end{equation}
Here we set $f_s(0)=0$ since $f_s(0)$ is a constant. For the second-order interaction $G_{ts}$, we let $x_j=0$ for $j\neq t, j\neq s$ in (\ref{Q:xs}):
\begin{equation*}
x_s G_s(x_s) + x_s x_t G_{ts} (x_t, x_s) = \theta_{st} x_t x_s + f_s(x_s).
\end{equation*}
Similarly, applying the previous argument on $P(\utilde{x})-P(\utilde{x}_{t}^0)$, we have
\begin{equation*}
x_t G_t(x_t) + x_s x_t G_{st} (x_s, x_t) = \theta_{ts} x_t x_s + f_t(x_t).
\end{equation*}
Therefore, if $\theta_{st}=\theta_{ts}$, then by \eqref{first.order},
\begin{equation*}
G_{st} (x_s, x_t) =G_{ts} (x_t, x_s)= \theta_{st}.
\end{equation*}
It is easy to show that, by setting $x_k=0$ $ (k\neq s, k\neq t, k\neq j)$ in \eqref{Q:xs}, the third-order interactions in $P(\utilde{x})$ are zero. Similarly, we can show that fourth-and-higher-order interactions are zero. Hence, we arrive at the following formula for $P$:
\begin{equation*}
P(\utilde{x})= \sum\limits_{s=1}^p  f_s(x_s) + \frac{1}{2}  \sum\limits_{s=1}^{p}\sum\limits_{t \neq s}\theta_{ts} x_s x_t.
\end{equation*}
Furthermore, $P(\utilde{x})=\log \{ h(\utilde{x}) /h(\utilde{0})\}$, so the function $h$ takes the form 
\begin{equation*}
h(\utilde{x}) \propto \exp \{P(\utilde{x})\} = \exp \left\{ \sum\limits_{s=1}^p  f_s(x_s)+\frac{1}{2}  \sum\limits_{s=1}^{p}\sum\limits_{t \neq s}	 \theta_{ts} x_s x_t \right\},
\end{equation*}
which is the same as \eqref{compat:proof}.
\end{proof}

\subsection{A Proof for Lemma~\ref{lmm:compat}}

\begin{proof} 
We first prove the claim about compatibility.

It is easy to verify that the conditional densities are integrable given the restrictions with asterisks in Table~\ref{tab1}. Therefore, these restrictions are sufficient for compatibility.

We now show that the restrictions with a dagger in Table~\ref{tab1}   are necessary, by investigating each of the distributions in Equations~\ref{cond:Gaussian} to \ref{cond:exp}. Note that we have limited our discussion to the case where all conditional densities are non-degenerate. Recall that we refer to the type of distribution of $x_s$ given the others as the node type of $x_s$. 

 Suppose that $x_s$ is exponential, as in \eqref{cond:exp}. By definition of the exponential distribution, it must be that $\eta_s=\alpha_{1s}+\sum_{t\neq s}\theta_{ts} x_t <0$. This leads to the following restrictions on $\theta_{ts}$: 1) When $x_t$ is Poisson or exponential, it must be that $\theta_{ts}\leq 0$ since $x_t$ is unbounded in $\mathcal{R}^{+}$.  2) When $x_t$ is Gaussian, then it must be that $\theta_{ts}=0$ since $x_t$ is unbounded on the real line.  3) Let  $I$ denote the indices of the Bernoulli variables. Then it must be that $\sum_{t \in I} |\theta_{ts}|< -\alpha_{1s}$ so that $\eta_s<0$ for any combination of $\{x_t\}_{t \in I}$. 

Suppose that $x_s$ is Gaussian, as in \eqref{cond:Gaussian}. Then $\alpha_{2s}$ has to be negative for the conditional density to be well-defined. 

Suppose that $x_s$ is Bernoulli or Poisson, as in Equations~\ref{cond:Bernoulli} or \ref{cond:Poisson}.  We can see that there are no restrictions on $\eta_s$, and thus no restrictions on $\theta_{ts}$ or $\alpha_{1s}$.

Hence, the conditions with a dagger in Table~\ref{tab1} are necessary for the conditional densities in Equations~\ref{cond:Gaussian} to \ref{cond:exp}  to be compatible.

We now show the statement about strong compatibility.

We first prove the necessity of the conditions in Table~\ref{tab1}. {Recall from Definition~\ref{compat:defn}  that in order for strong compatibility to hold, compatibility must hold, and any function $g$ that satisfies \eqref{g}  must be integrable.} 
 {Therefore, we derive} the necessary conditions for $g$ to be integrable.

For Gaussian nodes that are indexed by $J$, recall that $\Theta_{JJ}$ is defined as in \eqref{Thetajj}. Then, from properties of the multivariate Gaussian distribution, $\Theta_{JJ}$ must be negative definite if the joint density exists and is non-degenerate.

Let $x_1$ be a Poisson node, and $x_2$ an exponential node. Consider the ratio
\begin{equation*}
G(x_1, x_2) = \frac{g(x_1,x_2,0,...,0)}{g(0,0,0,...,0)} = \exp\{ -\log(x_1!) + \alpha_{11} x_1 +\theta_{12} x_1 x_2 + \alpha_{12} x_2\},
\label{ratio}
\end{equation*}
where $g$ is the function in \eqref{g}  . It is not hard to see that integrability of $G(x_1,x_2)$ is a necessary condition for integrability of the joint density. Summing over $x_1$ yields
\begin{equation*}
\sum\limits_{i=0}^{\infty}G(i,x_2)=\exp\{ \alpha_{12} x_2 + \exp(\alpha_{11}+\theta_{12}x_2)\}.
\end{equation*}
Therefore, if  $\sum_{i=0}^{\infty}G(i,x_2)$ is integrable with respect to the exponential node $x_2$, it must be the case that $\theta_{12}=\theta_{21} \leq 0$. Following a similar argument, the edge potential $\theta_{12}=\theta_{21}$ has to be non-positive when $x_2$  is Poisson, and zero when $x_2$  is Gaussian. 

A similar argument to the one just described can be applied to the exponential nodes. Such an argument reveals that conditions on the edge potentials of the exponential nodes that are necessary for $g$ to be a density are those stated in Table~\ref{tab1}.

For Bernoulli nodes,  no restrictions on the edge potentials are necessary in order for $g$ to be a density. 

Therefore, the conditions listed in Table~\ref{tab1}   are necessary for the conditional densities in Equations~\ref{cond:Gaussian} to \ref{cond:exp}  to be strongly compatible. 

We now show that the conditions listed in Table~\ref{tab1}  are sufficient for the conditional densities to be strongly compatible. We can restrict the discussion by conditioning on the Bernoulli nodes, since integrating over Bernoulli variables yields a mixture of finite components. Table~\ref{tab1}   guarantees that the Gaussian nodes are isolated from the Poisson and exponential nodes, as the corresponding edge potentials  are zero. From Table~\ref{tab1}, the distribution of Gaussian nodes is integrable, as $\Theta_{JJ}$ in \eqref{Thetajj}   is negative definite. Now we consider the Poisson and exponential nodes. For these,
\begin{equation*}
\exp \left\{ \sum\limits_{s=1}^p f_s(x_s) + \frac{1}{2}\sum\limits_{s=1}^{p}\sum\limits_{t \neq s} \theta_{st} x_s x_t  \right\} \leq  \exp \left\{ \sum\limits_{s=1}^p f_s(x_s)  \right\}
\end{equation*}
 since $\theta_{st} x_s x_t \leq 0$. So the joint density is dominated by the density of a model with no interactions, which is integrable since  $\alpha_{1t}$ for an exponential node $x_t$ is non-positive; this follows from the fact that $ 0 \leq \sum_{s \in I}|\theta_{st}| < - \alpha_{1t}$, as stated in Table~\ref{tab1}.  Therefore, the conditions listed in Table~\ref{tab1}  are also sufficient for the conditional densities in Equations~\ref{cond:Gaussian} to \ref{cond:exp}   to be strongly compatible.

\end{proof}

\subsection{A Proof for Theorem~\ref{thm}}\label{proof.thm}

\begin{proof}

Our proof is similar to that of Theorem~1 in \citet{yang2012}, and is based on the primal-dual witness method \citep{wainwright2009}. The primal-dual witness method studies the property of $\ell_1$-penalized estimators by investigating the sub-gradient condition of an oracle estimator. We assume that readers are familiar with the primal-dual witness method; for reference, see \citet{ravikumar2011} and \citet{yang2012}. Without loss of generality, we assume $s=p$ to avoid cumbersome notation. For other values of $s$, a similar proof holds with more complicated notation. Below we denote $\Theta_p$ as $\theta$, $\eta_p$ as $\eta$, and $\ell_p$ as $\ell$ for simplicity. We also denote the neighbours of $x_p$, $N(x_p)$, as $N$. 

The sub-gradient condition for \eqref{pl}   with respect to $(\theta^\T,\alpha_{1p})^\T$ is 
\begin{equation}
-\nabla\ell({\theta}, {\alpha}_{1p};  {X})+\lambda_n {Z}=0;  \quad {Z}_{t}= \text{sgn} ({\theta}_{t}) \quad \text{for} \  t < p; \quad {Z}_p=0,
\label{subgrad}
\end{equation}
where
 \begin{equation*}
 \text{sgn}(x)=\begin{cases}
 x/|x|, & x\neq 0,\\
\gamma \in [-1,1], & x=0.
\end{cases}
 \end{equation*}

We construct the oracle estimator $(\hat{\theta}_N^\T, \hat{\theta}_{\Delta}^\T, \hat{\alpha}_{1p})^\T$ as follows: first, let $\hat{\theta}_{\Delta}=0$ where $\Delta$ indicates the set of non-neighbours; second, obtain $\hat{\theta}_{N},\hat{\alpha}_{1p}$ by solving \eqref{pl}   with an additional restriction that $\hat{\theta}_{\Delta}=0$; third, set $\hat{Z}_{t}=\text{sgn}(\hat{\theta}_{t})$ for $t \in N$ and $\hat{Z}_p=0$; last, estimate $\hat{Z}_{\Delta}$ from \eqref{subgrad} by plugging in $\hat{\theta}, \hat{\alpha}_{1p}$ and $\hat{Z}_{\Delta^c}$. To complete the proof, we verify that $(\hat{\theta}_N^\T, \hat{\theta}_{\Delta}^\T, \hat{\alpha}_{1p})^\T$ and $\hat{Z}=(\hat{Z}_{N}^\T, \hat{Z}_{\Delta}^\T, 0)^\T$ is a primal-dual pair of \eqref{pl}   and recovers the true neighbourhood exactly.

 Applying the mean value theorem on each element of $\nabla \ell(\hat{\theta}, \hat{\alpha}_{1p};  {X} )$ in the subgradient condition \eqref{subgrad} gives
\begin{equation}
Q^*\bpm \hat{\theta}- {\theta}^* \\ \hat{\alpha}_{1p}-{\alpha}_{1p}^* \epm = -\lambda_n \hat{Z} +W^n+R^n,
\label{Taylorsubgrad}
\end{equation}
where $W^n = \nabla \ell (\theta^*, \alpha_{1p}^* ;  {X})$ is the sample score function evaluated at the true parameter $( {\theta^*}^\T, \alpha_{1p}^*)^{\T}$.  Recall that $Q^*=-\nabla^2 \ell(\theta^*, \alpha_{1p}^*; {X})$ is the negative Hessian of $\ell(\theta,\alpha_{1p}; {X})$ with respect to $(\theta^\T, \alpha_{1p})^\T$, evaluated at the true values of the parameters. In \eqref{Taylorsubgrad}, $R^n$ is the residual term from the mean value theorem, whose $k$th term is 
\begin{equation}
R^n_k = {[\nabla^2\ell(\bar{\theta}^k, \bar{\alpha}_{1p}^k ;  {X} )-\nabla^2 \ell(\theta^*, \alpha_{1p}^* ;  {X})]}^\T_k \bpm \hat{\theta}-\theta^* \\ \hat{\alpha}_{1p}-\alpha_{1p}^* \epm,
\label{eqn:R}
\end{equation}
where $\bar{\theta}^{k}$ denotes an intermediate point between ${\theta^*}$ and $ \hat{\theta}$, $\bar{\alpha}^{k}_{1p}$ denotes an intermediate point between $\alpha_{1p}^*$ and $\hat{\alpha}_{1p}$, and ${[\cdot]}^\T_k$ denotes the $k$th row of a matrix.

By construction, $\hat{\theta}_{\Delta}=0$. Thus, \eqref{Taylorsubgrad} can be rearranged as
\begin{equation}
\lambda_n \hat{Z}_{\Delta}=(W^n_{\Delta}+R^n_{\Delta})-Q^*_{\Delta \Delta^c}(Q^*_{\Delta^c \Delta^c})^{-1}(W^n_{\Delta^c}+R^n_{\Delta^c}-\lambda_n \hat{Z}_{\Delta^c}).
\label{strictd.eq}
\end{equation}
We obtain an estimator $\hat{Z}_{\Delta}$ by plugging $\hat{\theta}, \hat{\alpha}_{1p}$ and $\hat{Z}_{\Delta^c}$ into \eqref{strictd.eq}. To complete the proof, we need to verify strict dual feasibility,
\begin{equation}
\| \hat{Z}_{\Delta}\|_{\infty} < 1,
\label{strictdf}
\end{equation}
and sign consistency,
\begin{equation}
\text{sgn}(\hat{\theta}_{t})=\text{sgn}(\theta_{t}^*) \quad \text{for any} \  t \in N.
\label{signc}
\end{equation}

In \eqref{strictd.eq}, $ \max_{l\in \Delta} \| Q^*_{l\Delta^c} (Q^*_{\Delta^c \Delta^c})^{-1}\|_1 \leq 1-a$ by Assumption~\ref{irrep}. The following lemmas characterize useful concentration inequalities regarding $W^n$, $R^n$, and $\hat{\theta}_{N} - \theta^*_{N}$. Proofs of Lemmas~\ref{lemma.W} and \ref{lemma.R} are given in Sections~\ref{lemmaW.proof} and \ref{lemmaR.proof}, respectively. \\

\begin{lemma}
Suppose that 
$$ \frac{8(2-a)}{a}\{\delta_2 \kappa_2 \log (2p)/n\}^{1/2} \leq \lambda_n \leq \frac{2(2-a)}{a} \delta_2  \kappa_2 M, $$
where $\delta_2$ is defined in Proposition~\ref{prop.e2}, and $a$ and $\kappa_2$ are defined in Assumptions~\ref{irrep} and \ref{D}, respectively. 
 Then,
\begin{equation*}
pr\left( \left.\|W^n\|_\infty>\frac{a \lambda_n }{8-4a} \right| \xi_2,\xi_1 \right) \leq \exp (-c_3 \delta_3 n),
\end{equation*}
 where $\delta_3=1/(\kappa_2 \delta_2)$ and $c_3$ is some positive constant.
\label{lemma.W}
\end{lemma}

\begin{lemma}
Suppose that $\xi_1$ and $\|W^n\|_\infty \leq a\lambda_n /(8-4a) $ hold and
$$ \lambda_n \leq \min\left\{ \frac{a \Lambda_{1}^2 (d+1)^{-1}}{288 (2-a) \kappa_2 \Lambda_{2} },  \frac{\Lambda_{1}^2 (d+1)^{-1}}{12 \Lambda_{2} \kappa_3 \delta_1 \log p} \right\}, $$
where $\delta_1$ is defined in Proposition~\ref{prop.e1}, and $a$ and $\kappa_3$ are defined in Assumptions~\ref{irrep} and \ref{D}, respectively.
 Then with probability 1,
$$ \| \hat{\theta}_{N}-\theta^*_{N}\|_2 < \frac{10}{\Lambda_{1}}(d+1)^{1/2}\lambda_n, \quad \|R^n\|_\infty \leq \frac{a \lambda_n }{8-4a}.$$
\label{lemma.R}
\end{lemma}

We now continue with the proof of Theorem~\ref{thm}. Given Assumption~\ref{tuning.range}, the conditions regarding $\lambda_n$ are met for Lemmas~\ref{lemma.W} and \ref{lemma.R}. 

We now assume that $\xi_1, \xi_2$ and the event $\|W^n\|_\infty \leq a\lambda_n /(8-4a)$ are true so that the conditions for the two lemmas are satisfied. We derive the lower bound for the probability of these events at the end of the proof. 

First, applying Lemma~\ref{lemma.R} and Assumption~\ref{irrep} to \eqref{strictd.eq} yields
\begin{equation}
\begin{aligned}
\|\hat{Z}_{\Delta}\|_{\infty} \leq & \max_{l \in \Delta} \| Q^*_{l \Delta^c} (Q^*_{\Delta^c \Delta^c})^{-1}\|_1 \left( \|W^n_{\Delta^c}\|_{\infty} +\|R^n_{\Delta^c}\|_{\infty}+\lambda_n \|\hat{Z}_{\Delta^c}\|_{\infty} \right)/\lambda_n+\\
& \left( \|W^n_{\Delta}\|_{\infty}+\|R^n_{\Delta}\|_{\infty} \right)/\lambda_n\\
\leq & (1-a)+(2-a) \left\{\frac{a}{4(2-a)}+\frac{a}{4(2-a)} \right\} < 1.
\label{strictdf.result}
\end{aligned}
\end{equation}

Next, applying Lemma~\ref{lemma.R} and a norm inequality to $\| \hat{\theta}_{N}-\theta^*_{N}\|_{\infty}$ gives
\begin{equation}
\| \hat{\theta}_{N}-\theta^*_{N}\|_{\infty} \leq \| \hat{\theta}_{N}-\theta^*_{N}\|_2 < \frac{10}{\Lambda_{1}} (d+1)^{1/2} \lambda_n \leq \min_{t} |\theta_{t}|,
\label{signcrt.result}
\end{equation}
since $\min_t |\theta_{t}| \geq 10 (d+1)^{1/2} \lambda_n/\Lambda_{1} $ by Assumption~\ref{thetamin}. The strict inequality in \eqref{signcrt.result} ensures that the sign of the estimator is consistent with the sign of the true value for all edges. 

Equations \ref{strictdf.result} and \ref{signcrt.result} are sufficient to establish the result, i.e., $\hat{N}=N$. Let $A$ be the event $\|W^n\|_\infty \leq a\lambda_n /(8-4a) $. Recall that we have assumed events $A$, $\xi_1$, and $\xi_2$ to be true in order to prove \eqref{strictdf.result} and \eqref{signcrt.result}. We now derive the lower bound for the probability of $A\cap \xi_1 \cap \xi_2$. 

 Using the fact that
$$\text{pr}\{ (A\cap \xi_1 \cap \xi_2)^c \}\leq \text{pr}(A^c \mid \xi_1\cap \xi_2) + \text{pr}\{(\xi_1 \cap \xi_2)^c\} \leq \text{pr}(A^c \mid \xi_1, \xi_2)+\text{pr}(\xi_1^c)+\text{pr}(\xi_2^c),$$
we know the probability of $A\cap \xi_1 \cap \xi_2$ satisfies
\begin{equation*}
\text{pr}\left\{ \left(\|W^n\|_\infty \leq \frac{a}{2-a}\frac{\lambda_n}{4}\right)\cap \xi_2\cap \xi_1 \right\}\geq 1-c_1 p^{-\delta_1+2}-\exp(-c_2 \delta_2^2 n)-\exp(-c_3 \delta_3 n),
\end{equation*}
where $c_1$, $c_2$, and $c_3$ are constants from Proposition~\ref{prop.e1}, Proposition~\ref{prop.e2}, and Lemma~\ref{lemma.W}. Thus, the event $A\cap \xi_1 \cap \xi_2$ happens with high probability when the sample size $n$ is large. This completes the proof.
\end{proof}

\subsection{A Proof for Lemma \ref{lemma.W}}\label{lemmaW.proof}
\begin{proof}
Recall that $\eta^{(i)} = \alpha_{1p} + \sum_{t<p} \theta_{t} x_t^{(i)}$ and that we have assumed that $\alpha_{kp}$ is known for $k\geq 2$. We can rewrite the conditional density in \eqref{cond}   as
$$ p(x_p \mid x_{-p}) \propto \exp \{ \eta x_p - D(\eta) \}. $$ 
For any $t<p$, 
 \begin{equation}
W^n_t=\frac{\partial \ell}{\partial \theta_{t}}= \sum\limits_{i=1}^{n} \frac{\partial \ell}{\partial \eta^{(i)}} \frac{\partial \eta^{(i)}}{\partial \theta_{t}} = \frac{1}{n} \sum\limits_{i=1}^{n} \{x_p^{(i)}-D^{'}(\eta^{(i)})\}x_t^{(i)}.
\label{eqn:W}
\end{equation}

Recall that $M$ is a large constant introduced in Assumption~\ref{D}. Suppose that $M$ is sufficiently large that $|\alpha_{1p}^*|+\sum_{ k <p} | \theta_{k}^*| < M/2$. For every $v$ such that $0<v<M/2$, 
 \begin{equation}
\begin{aligned}
E \left( \left. \exp \left[ v x_t^{(i)}\left\{ x_p^{(i)}-D^{'}(\eta^{(i)})\right\}\right] \right|  {X}_{- p}\right) = &
 E \left\{ \left. \exp \left( v x_t^{(i)} x_p^{(i)}\right) \right|  {X}_{- p}\right\}  \exp \left\{- v x_t^{(i)}D^{'}(\eta^{(i)})\right\}  \\
= & \exp \left\{ D(\eta^{(i)}+v x_t^{(i)})- D(\eta^{(i)}) \right\} \exp \left\{- v x_t^{(i)}D^{'}(\eta^{(i)})\right\} \\
= & \exp \left\{ v x_t^{(i)} D^{'}(\eta^{(i)})+  (v x^{(i)}_t )^2 \frac{D^{''}(\tilde{\eta})}{2}  \right\} \exp \left\{- v x_t^{(i)}D^{'}(\eta^{(i)})\right\}   \\
= &\exp \left\{ (v x^{(i)}_t )^2 \frac{D^{''}(\tilde{\eta})}{2} \right\}, \ \ \tilde{\eta} \in [\eta^{(i)}, \eta^{(i)}+  v x_t^{(i)} ],  \\
\end{aligned}
\label{eqn:W.t}
\end{equation}
where the second equality was derived using the properties of the moment generating function of the exponential family, and the third equality follows from a second-order Taylor expansion.
 Since $\tilde{\eta} \in [\eta^{(i)}, \eta^{(i)}+ v x_t^{(i)} ]$, the event $\xi_1$ implies that 
\begin{equation}
|\tilde{\eta}|\leq |\alpha_{1p}^*|+\sum_{ k <p} |x^{(i)}_{k} \theta_{k}^*| + |v x_t^{(i)}| \leq  |\alpha_{1p}^*|+ (\sum_{ k <p} |\theta_{k}^*| + |v| )\underset{t,i}{\max} |x_t^{(i)}|  \leq M \delta_1 \log p. 
\label{eqn:eta.max}
\end{equation}
Therefore, the condition of Assumption~\ref{D} is satisfied, and thus $|D^{''}(\tilde{\eta})|\leq \kappa_2$. Recalling that $\{x^{(i)}\}^n_{i=1}$ are independent samples, it follows from \eqref{eqn:W} and \eqref{eqn:W.t} that 
\begin{equation}
\begin{aligned}
E\left\{  \exp ( v n W^n_t)  \mid \xi_2, \xi_1 \right\} = &E \left[ E \left\{ \exp ( v n W^n_t) \mid X_{-p}, \xi_2, \xi_1 \right\} \mid \xi_2, \xi_1 \right] \\
\leq & E \left[ \left. \exp \left\{v^2 \frac{\kappa_2}{2}\sum\limits_{i=1}^{n}(x^{(i)}_t)^2 \right\} \right| \xi_2, \xi_1 \right]  \\
\leq & \exp ( n v^2 \kappa_2 \delta_2/2  ),
\end{aligned}
\label{eqn:W.pos}
\end{equation}
where we use the event $\xi_2$ in the last inequality. Similarly, 
\begin{equation}
E \left\{\exp ( -n vW^n_t)\mid  \xi_2,\xi_1\right\} \leq \exp ( n v^2 \kappa_2 \delta_2/2 ).
\label{eqn:W.neg}
\end{equation}

Furthermore, one can see from a similar argument as in \eqref{eqn:W} and \eqref{eqn:W.t} that
\begin{equation*}
\begin{aligned}
E \left\{\exp ( n vW^n_p) \mid \xi_1 \right\} = & E \left\{ \left. \exp \left(v n \frac{\partial \ell}{\partial \alpha_{1p}} \right)\right| \xi_1  \right\}  \\
= & \prod_{i=1}^n E \left(\exp [ \left.v \{x_p^{(i)}-D^{'}(\eta^{(i)}) \}]\right|  \xi_1 \right) \leq \exp ( n \kappa_2v^2/2 ) .
\end{aligned}
\end{equation*}
 
 We focus on the discussion of \eqref{eqn:W.pos} and \eqref{eqn:W.neg} since $\delta_2 \geq 1$. For some $\delta$ to be specified, we let $v=\delta/(\kappa_2 \delta_2)$ and apply the Chernoff bound \citep{chernoff1952, ravikumar2004} with \eqref{eqn:W.pos} and \eqref{eqn:W.neg} to get 
 
\begin{equation*}
\text{pr}(|W^n_t|>\delta \mid \xi_2, \xi_1) \leq \frac{ E \{\exp( v n W^n_t )  \mid \xi_2, \xi_1\} }{\exp(v n \delta)} + \frac{ E \{\exp( -v n W^n_t) \mid \xi_2, \xi_1 \} }{\exp(v n \delta)} \leq  2\exp \left( -n\frac{\delta^2}{2\kappa_2 \delta_2} \right).
\end{equation*}
Letting $\delta=a\lambda_n /(8-4a) $  and using the Bonferroni inequality, we get
 \begin{equation}
\begin{aligned}
\text{pr}\left( \left.\|W^n\|_\infty > \frac{a}{2-a} \frac{\lambda_n}{4} \right| \xi_2,\xi_1 \right) \leq & 2\exp \left\{ -n\frac{a^2 \lambda_n^2}{32 (2-a)^2\kappa_2 \delta_2} +\log (p) \right\}\\
  \leq &  \exp \left\{ - \frac{ a^2 \lambda_n^2}{64 (2-a)^2\kappa_2 \delta_2} n\right\}=\exp(-c_3 \delta_3 n),
  \end{aligned}
 \label{W}
\end{equation}
where $\delta_3=1/(\kappa_2 \delta_2)$ and $c_3= a^2 \lambda_n^2/\{64 (2-a)^2\}$. In \eqref{W}, we made use of the assumption that $\lambda_n \geq 8(2-a)\{\kappa_2 \delta_2 \log (2p)/n \}^{1/2} /a$, and we also require that
$\lambda_n \leq 2(2-a) \kappa_2 \delta_2 M/a $ since $v=a\lambda_n/ \{ (8-4a) \kappa_2 \delta_2 \}\leq M/2$.
\end{proof}

\subsection{A Proof for Lemma \ref{lemma.R}}\label{lemmaR.proof}

\begin{proof}
We first prove that $ \| \hat{\theta}_{N}-\theta^*_N\|_2 < 10(d+1)^{1/2}\lambda_n/\Lambda_{1}.$ 

Following the method in \citet{fan2004} and \citet{ravikumar2010}, we construct a function $F(u)$ as 
\begin{equation}
F(u) =- \ell (\theta^*+u_{-p}, \alpha_{1p}^*+u_p;  {X})+\ell(\theta^*,\alpha_{1p}^*; {X})+\lambda_n\|\theta^*+u_{-p}\|_1-\lambda_n\|\theta^*\|_1,
\label{defineF}
\end{equation} 
where $u$ is a $p$-dimensional vector and $u_{\Delta}=0 $. $F(u)$ has some nice properties: (i) $F(0)=0$ by definition; (ii) $F(u)$ is convex in $u$ given the form of \eqref{cond}  ; and (iii)  by the construction of the oracle estimator $\hat{\theta}$, $F(u)$ is minimized by $\hat{u}$ with $\hat{u}_{-p} = \hat{\theta}-\theta^*$ and $\hat{u}_p = \hat{\alpha}_{1p}-\alpha_{1p}^*$. 

We claim that if there exists a constant $B$ such that $F(u)>0$ for any $u$ such that $\|u\|_2=B$ and $u_{\Delta}=0$, then $\| \hat{u} \|_2 \leq B$.  To show this, suppose that $\| \hat{u} \|_2  > B$ for such a constant. Let $t=B/\|\hat{u}\|_2$. Then, $t<1$, and the convexity of $F(u)$ gives 
\begin{equation*}
F(t \hat{u}) \leq (1-t) F(0)+tF(\hat{u})\leq 0.
\end{equation*}
Thus, $\|t\hat{u}\|_2=B$ and $(t\hat{u})_{\Delta}=t\hat{u}_{\Delta}=0$, but $F(t \hat{u})\leq 0$, which is a contradiction.

Applying a Taylor expansion to the first term of $F(u)$ gives 
\begin{equation*}
\begin{aligned}
F(u)=& -{\nabla \ell (\theta^*, \alpha_{1p}^*; {X})}^\T u- {u}^\T  \nabla^2\ell(\theta^*+v u_{-p}, \alpha_{1p}^*+v u_p; X) u/2 +\lambda_n (\|\theta^*+u_{-p}\|_1-\|\theta^*\|_1 ) \\
= & \ \text{I}+\text{II}/2+\text{III} ,
\end{aligned}
\end{equation*} 
for some $v \in [0,1]$. Recall that $u_{\Delta}=0$ as defined in \eqref{defineF}. The gradient and Hessian are with respect to the vector $(\theta^T, \alpha_{1p})^\T$. 

We now proceed to find a $B$ such that for $\|u\|_2=B$ and $u_{\Delta}=0$, the function $F(u)$ is always greater than 0. First, given that $\|W^n\|_\infty \leq a \lambda_n / (8-4a) $ and $a < 1$ assumed in Assumption~\ref{irrep},
\begin{equation*}
|\text{I}|=|{(W^n)}^\T u|\leq \| W^n\|_\infty \|u\|_1 \leq \frac{a}{2-a} \frac{\lambda_n}{4} (d+1)^{1/2} B \leq \frac{\lambda_n}{4} (d+1)^{1/2} B.
\end{equation*}
Next, by the triangle inequality and the Cauchy-Schwarz inequality,
\begin{equation*}
\text{III}\geq -\lambda_n\|u_{-p}\|_1 \geq -\lambda_n d^{1/2}  \|u_{-p}\|_2 \geq -\lambda_n (d+1)^{1/2} B.
\end{equation*}

To bound II, we note that 
\begin{equation*}
- \nabla^2\ell(\theta^*+vu_{-p}, \alpha_{1p}^*+vu_p;X)=\frac{1}{n} \sum\limits_{i=1}^{n} \utilde{x}^{(i)}_{0}  {(\utilde{x}^{(i)}_{0})}^\T D^{''}(\eta^{(i)}_r),
\end{equation*}
where $x_0=(x_{-p}^{\T}, 1)^{\T}$ as in Assumption~\ref{dep}, and $\eta^{(i)}_r=\alpha_{1p}^*+v u_p + \sum_{t<p} (\theta_t^*+v u_t)x_t^{(i)}$. Applying a  Taylor expansion on each $D^{''}(\eta_r^{(i)})$ at $\eta^{(i)}=\alpha_{1p}^*+\sum_{t<p} \theta_t^* x_t^{(i)}$, we get
\begin{equation*}
\begin{aligned}
- \nabla^2\ell(\theta^*+vu_{-p}, \alpha_{1p}^*+vu_p; X)=& \frac{1}{n} \sum\limits_{i=1}^{n}  \utilde{x}^{(i)}_{0} {(\utilde{x}^{(i)}_{0})}^\T D^{''}(\eta^{(i)}) +\frac{1}{n} \sum\limits_{i=1}^{n}    \utilde{x}^{(i)}_{0} {(\utilde{x}^{(i)}_{0})}^\T D^{'''}(\tilde{\eta}^{(i)}) \left( v{u}^\T\utilde{x}^{(i)}_{0} \right)\\
= & Q^* + \frac{1}{n} \sum\limits_{i=1}^{n}    \utilde{x}^{(i)}_{0} {(\utilde{x}^{(i)}_{0})}^\T D^{'''}(\tilde{\eta}^{(i)}) \left( v{u}^\T\utilde{x}^{(i)}_{0} \right),
\end{aligned}
\end{equation*}
where $\tilde{\eta}^{(i)} \in [\eta^{(i)}, \eta^{(i)}_r]$. Using the argument on $\tilde{\eta}$ in \eqref{eqn:eta.max} and the fact that $v\leq 1$ and $ \|u\|_2=B$, we can see that $\tilde{\eta}^{(i)}$ is in the range required for Assumption~\ref{D} to hold given $\xi_1$.  
Therefore, applying Assumption~\ref{D} we can write 
\begin{equation*}
\begin{aligned}
\text{II} \geq &  \min_{u: \|u\|_2=B, u_{\Delta}=0}  \{ - u^T \nabla^2 \ell(\theta^*+vu_{-p}, \alpha_{1p}^*+vu_p; X) u \} \\
\geq & B^2 \Lambda_{\min} (Q^*_{\Delta^c \Delta^c} )-
\max_{v \in [0,1]}\max_{u: \|u\|_2=B , u_{\Delta}=0} {u}^\T \left\{ \frac{1}{n}\sum\limits_{i=1}^{n} D^{'''}(\tilde{\eta}^{(i)}) (v{u}^\T\utilde{x}^{(i)}_{0} ) \utilde{x}^{(i)}_{0} ({\utilde{x}^{(i)}_{0})}^\T \right\} u \\
\geq & \Lambda_{1} B^2 -
\max_{u: \|u\|_2=B , u_{\Delta}=0} \left\{\max_{i, v\in [0,1]}( v {u}^\T\utilde{x}^{(i)}_{0}) \max_{\tilde{\eta}^{(i)}} D^{'''}(\tilde{\eta}^{(i)})  \frac{1}{n}\sum\limits_{i=1}^{n} ({u}^\T\utilde{x}^{(i)}_{0})^2 \right\}  \\
 \geq & \Lambda_{1} B^2 - \kappa_3 \max_{i, u: \|u\|_2=B , u_{\Delta}=0, v \in [0,1]} ( v {u}^\T\utilde{x}^{(i)}_{0} )  \max_{u: \|u\|_2=B , u_{\Delta}=0}  \left\{ \frac{1}{n} \sum\limits_{i=1}^{n}  ({u}^\T\utilde{x}^{(i)}_{0})^2 \right\}. \quad          \\
\end{aligned}
\end{equation*}
By inspection, the maximum of $u^\T x^{(i)}_0$ is non-negative. Thus, the maximum of $ vu^T \utilde{x}^{(i)}_{0}$ is achieved at $v=1$. Then, using $\xi_1$ and Assumption~\ref{dep},
\begin{equation*}
\begin{aligned}
\text{II} \geq & \Lambda_{1} B^2  - \kappa_3  B (d+1)^{1/2} \delta_1 \log (p)  B^2 \Lambda_{\max} \left\{ \frac{1}{n}\sum\limits_{i=1}^{n}  \utilde{x}^{(i)}_{0} {(\utilde{x}^{(i)}_{0})}^\T \right\} \\
\geq & \Lambda_{1} B^2 - \kappa_3  B^3  (d+1)^{1/2}\delta_1 \log (p)  \Lambda_{2}.
\end{aligned}
\end{equation*} 
Thus, if our choice of $B$ satisfies
\begin{equation}
\Lambda_{1}  - \delta_1 \log (p) B \kappa_3 (d+1)^{1/2}  \Lambda_{2} \geq \frac{\Lambda_{1}}{2},
\label{lambdan.cond1}
\end{equation}
then the lower bound of $F(u)$ is
\begin{equation*}
F(u) \geq  -\frac{\lambda_n}{4} (d+1)^{1/2} B +\frac{\Lambda_{1}}{4} B^2-\lambda_n (d+1)^{1/2} B.
\end{equation*}
So, $F(u)>0$ for any $B>5 (d+1)^{1/2}\lambda_n /\Lambda_{1}$. We can hence let 
\begin{equation}
B= 6(d+1)^{1/2}\lambda_n/\Lambda_{1}
\label{eqn:B}
\end{equation}
 to get
\begin{equation}
 \| \hat{\theta}_N-\theta^*_N\|_2 \leq \| \hat{u}\|_2 \leq B = \frac{6}{\Lambda_{1}}(d+1)^{1/2}\lambda_n.
 \label{eqn:uB}
\end{equation}
And thus, $\| \hat{\theta}_N-\theta^*_N\|_2 < 10\lambda_n (d+1)^{1/2} / \Lambda_{1}$. It is easy to show that \eqref{eqn:B} satisfies \eqref{lambdan.cond1} provided that 
$$\lambda_n \leq \frac{\Lambda_{1}^2(d+1)^{-1}}{12 \Lambda_{2} \kappa_3 \delta_1 \log p}.$$

To find the bound for $R^n$ defined in \eqref{eqn:R}, we first recall that $(\bar{\theta}^{\T},\bar{\alpha})^{\T}$ is an intermediate point between $( {\theta^*}^\T, \alpha_{1p}^*)^{\T}$ and $( \hat{\theta}^{\T}, \hat{\alpha}_{1p})^{\T}$. We denote $\bar{\eta}^{(i)}=\bar{\alpha}_{1p}+\sum_{t < p} \bar{\theta}_{t} x_t^{(i)}$, and observe that $|\bar{\eta}^{(i)}| \leq M  \delta_1 \log p $ for $i=1, \ldots, n$ using the argument of \eqref{eqn:eta.max}, which implies that Assumption~\ref{D} is applicable. Thus, 
\begin{equation*}
 \begin{aligned}
 \Lambda_{\max}\{\nabla^2\ell(\bar{\theta}, \bar{\alpha}_{1p};X )-\nabla^2 \ell(\theta^*, {\alpha}_{1p}^*;X)\}=&\underset{\|u\|_2=1}{\max} {u}^\T \{ \nabla^2\ell(\bar{\theta}, \bar{\alpha}_{1p};X )-\nabla^2 \ell(\theta^*, {\alpha}_{1p}^*;X)\}u \\
= & \underset{\|u\|_2=1}{\max} {u}^\T \left[ \frac{1}{n} \sum\limits_{i=1}^{n} \left\{ D^{''}(\bar{\eta}^{(i)} ) - D^{''} (\eta^{*}) \right\}\utilde{x}^{(i)}_{0} {(\utilde{x}^{(i)}_{0})}^\T \right] u.
\end{aligned}
\end{equation*}
By Assumption 3, $|D^{''}(\bar{\eta}^{(i)} ) - D^{''} (\eta^{*})| \leq 2 \kappa_2 $, and so
\begin{equation*}
 \begin{aligned}
 \Lambda_{\max}\{ \nabla^2\ell(\bar{\theta}, \bar{\alpha}_{1p};X )-\nabla^2 \ell(\theta^*, {\alpha}_{1p}^*;X) \} = & \underset{\|u\|_2=1}{\max} {u}^\T \left[ \frac{1}{n} \sum\limits_{i=1}^{n} \left\{ D^{''}(\bar{\eta}^{(i)} ) - D^{''} (\eta^{*})\right\} \utilde{x}^{(i)}_{0} {(\utilde{x}^{(i)}_{0})}^\T \right] u \\
 \leq & 2\kappa_2 \underset{\|u\|_2=1}{\max} {u}^\T \left\{ \frac{1}{n} \sum\limits_{i=1}^{n}  \utilde{x}^{(i)}_{0} {(\utilde{x}^{(i)}_{0})}^\T \right\} u  \leq  2\kappa_2  \Lambda_{2},
\end{aligned}
\end{equation*}
using Assumption~\ref{dep} at the last inequality. Hence, we arrive at
\begin{equation*}
 \begin{aligned}
 \|R^n\|_{\infty} \leq& \|R^n\|_2^2=\left\|\{\nabla^2\ell(\bar{\theta}, \bar{\alpha}_{1p};X )-\nabla^2 \ell(\theta^*, {\alpha}_{1p}^*;X)\}^\T \bpm \hat{\theta}-\theta^* \\ \hat{\alpha}_{1p}-{\alpha}_{1p}^*  \epm \right\|_2^2  \\
 \leq &\Lambda_{\max}\{\nabla^2\ell(\bar{\theta}, \bar{\alpha}_{1p};X )-\nabla^2 \ell(\theta^*, {\alpha}_{1p}^*;X) \} \left\| \bpm \hat{\theta}-\theta^* \\ \hat{\alpha}_{1p}-{\alpha}_{1p}^*  \epm \right\|_2^2 \\
 = &\Lambda_{\max}\{\nabla^2\ell(\bar{\theta}, \bar{\alpha}_{1p};X )-\nabla^2 \ell(\theta^*, {\alpha}_{1p}^*;X) \} \left\|  \hat{u} \right\|_2^2 \\
 \leq & \frac{72\kappa_2 \Lambda_{2} }{\Lambda^2_{1}} (d+1) \lambda_n^2,
 \end{aligned}
\end{equation*}
where the last inequality follows from \eqref{eqn:uB}. So $\|R^n\|_\infty \leq a\lambda_n /(8-4a)$ if
\begin{equation}
 \lambda_n \leq \frac{a}{2-a} \frac{\Lambda_{1}^2}{288 (d+1)\kappa_2\Lambda_{2} },
  \label{lambdan.cond2}
\end{equation}
which holds by assumption. 
\end{proof}

\subsection{A Proof for Corollary~\ref{col1}}\label{col1.proof}

\begin{proof}

The proof is essentially the same as the proof in Section~\ref{proof.thm}. We first show that a modified version of Lemma~\ref{lemma.R} holds with fewer conditions.

\begin{lemma}
Suppose that {$p(x_p|x_{-p})$} follows a Gaussian distribution as in \eqref{cond:Gaussian}  , and $\|W^n\|_\infty \leq a\lambda_n /(8-4a) $. Then
$$ \| \hat{\theta}_{N}-\theta^*_{N}\|_2 < \frac{10}{\Lambda_{1}}(d+1)^{1/2}\lambda_n, \quad \|R^n\|_\infty =0.$$
\label{lemma.R2}
\end{lemma}
\begin{proof} To prove this lemma, we go through the argument in Section~\ref{lemmaR.proof}. But for $\text{II}$ we note that
\begin{equation*}
\begin{aligned}
\text{II} \geq &  \min_{u: \|u\|_2=B, u_{\Delta}=0}  \{ - u^T \nabla^2 \ell(\theta^*+vu_{-p}, \alpha_{1p}^*+vu_p; X) u \} \\
\geq & B^2 \Lambda_{\min} (-Q^*_{\Delta^c \Delta^c} )-
\max_{v \in [0,1]}\max_{u: \|u\|_2=B , u_{\Delta}=0} {u}^\T \left\{ \frac{1}{n}\sum\limits_{i=1}^{n} D^{'''}(\tilde{\eta}^{(i)}) (v{u}^\T\utilde{x}^{(i)}_{0} ) \utilde{x}^{(i)}_{0} ({\utilde{x}^{(i)}_{0})}^\T \right\} u \\
\geq & \Lambda_{1} B^2 -0,
\end{aligned}
\end{equation*} 
since $D^{'''}(\tilde{\eta}^{(i)})=0$ for a Gaussian distribution. Therefore, 
\begin{equation*}
F(u) \geq   -\frac{\lambda_n}{4} (d+1)^{1/2} B +\frac{1}{2}\Lambda_{1} B^2-\lambda_n (d+1)^{1/2} B.
\end{equation*}
So, $F(u)> 0$ for $B > 5\lambda_n (d+1)^{1/2} /(2\Lambda_{1})$. We can hence let $B= 5(d+1)^{1/2}\lambda_n/\Lambda_{1}$ to get
\begin{equation*}
 \| \hat{\theta}_N-\theta^*_N\|_2 \leq \| \hat{u}\|_2 \leq B= \frac{5}{\Lambda_{1}}(d+1)^{1/2}\lambda_n.
\end{equation*}
Thus, $\|\hat{\theta}_N - \theta^{*}_N \|_2<10 \lambda_n (d+1)^{1/2}/\Lambda_1$.  And $\|R^n\|_\infty =0$ trivially as $D^{''}(\bar{\eta}^{(i)} ) - D^{''} (\eta^{*})=0$ for a Gaussian distribution. \end{proof}

With Lemma~\ref{lemma.R2}, we can then verify \eqref{strictdf.result} and \eqref{signcrt.result} as in Section~\ref{proof.thm}. Finally, we drop the requirement of $\xi_1$ in the condition of Lemma~\ref{lemma.R2}, so the probability of $\hat{N}=N$ is
\begin{equation*}
\text{pr}\left\{ \left(\|W^n\|_\infty \leq \frac{a}{2-a}\frac{\lambda_n}{4}\right)\cap \xi_2\right\}\geq 1-\exp(-c_2 \delta_2^2 n)-\exp(-c_3 \delta_3 n),
\end{equation*}
where $c_2$ and $c_3$ are constants from Proposition~\ref{prop.e2}  and Lemma~\ref{lemma.W}. 
\end{proof}

\subsection{Additional Details of Data-Generation Procedure} \label{generation_detail}

Here we provide additional details of  the data-generation procedure described in Section~\ref{generate}. In particular, we describe the approach used to guarantee that   the conditions listed in Table~\ref{tab1}  for strong compatibility of the conditional distributions are satisfied.

Recall from Table~\ref{tab1} that in order for strong compatibility to hold,  the matrix $\Theta_{JJ}$ in \eqref{Thetajj} that contains the edge potentials between the Gaussian nodes must be negative definite. If  $\Theta_{JJ}$ generated as described in Section~\ref{generate}   is not negative definite, then we define a matrix $T_{JJ}$ as $$T_{JJ}=   -{ \Theta}_{JJ}+ \left\{ \Lambda_{\min}({\Theta}_{JJ}) - 0.1 \right\} {  I} ,$$ where $\Lambda_{\min}({\Theta}_{JJ})$ denotes the minimum eigenvalue of $\Theta_{JJ}$. Thus, $T_{JJ}$ is guaranteed to be negative definite, as all its eigenvalues are no larger than $-0.1$. We then standardize $T_{JJ}$ so that its diagonal elements equal $-1$,
 $$\tilde{T}_{JJ}=\text{diag}(|T_{11}|^{-1/2}, \ldots,|T_{mm}|^{-1/2}  )\   T_{JJ}\ \text{diag}(|T_{11}|^{-1/2}, \ldots,|T_{mm}|^{-1/2}  ).$$
Finally, we replace $\Theta_{JJ}$ with $\tilde{T}_{JJ}$.

Table~\ref{tab1}   also indicates that for strong compatibility to hold, the edge potential between two Poisson nodes must be negative. Therefore, after generating edge potentials as described in Section~\ref{generate}, we replace
 $\theta_{st}$ with $-|\theta_{st}|$ where $x_s$ and $x_t$ are Poisson nodes.

\end{document}